\documentclass[poms,nonblindrev]{poms1} 

\OneAndAHalfSpacedXI 

\usepackage[center]{caption}
\usepackage{threeparttable}
\usepackage{hyperref}
\hypersetup{
    colorlinks=true,
    linkcolor=blue,
    filecolor=gray,      
    urlcolor=blue,
    citecolor=blue,
}
\usepackage{graphicx}
\usepackage{epstopdf}
\usepackage{epsfig}
\usepackage{natbib}
\usepackage{xcolor}
\usepackage{amsmath}
\usepackage{booktabs}

 \bibpunct[, ]{(}{)}{,}{a}{}{,}%
 \def\bibfont{\small}%
\TheoremsNumberedThrough     
\ECRepeatTheorems

\EquationsNumberedThrough    

\begin{document}


\RUNAUTHOR{Yang, Lee, and Chen}

\RUNTITLE{Competitive Demand Learning: A Non-cooperative Pricing Algorithm with Coordinated Price Experimentation}

\TITLE{Competitive Demand Learning: A Non-cooperative Pricing Algorithm with Coordinated Price Experimentation}


\ARTICLEAUTHORS{%
\AUTHOR{Yongge Yang}
\AFF{Department of Industrial Engineering and Engineering Management, National Tsing Hua University, Taiwan. \EMAIL{s106034466@m106.nthu.edu.tw}, \URL{}}
\AUTHOR{Yu-Ching Lee}
\AFF{Corresponding. Department of Industrial Engineering and Engineering Management, National Tsing Hua University, Taiwan. \EMAIL{yclee@ie.nthu.edu.tw}, \URL{}}
\AUTHOR{Po-An Chen}
\AFF{Co-Corresponding. Institute of Information Management, National Yang Ming Chiao Tung University, Taiwan. \EMAIL{poanchen@nycu.edu.tw}, \URL{}}
}
\ABSTRACT{%
We consider a periodical equilibrium  pricing  problem  for multiple  firms  over  a  planning  horizon  of $T$ periods.  At each period, firms set their selling prices and receive stochastic demand from consumers. Firms do not know their underlying demand curve, but they wish to determine the selling prices to maximize total revenue under competition. Hence, they have to do some price experiments such that the observed demand data are informative to make price decisions. However, uncoordinated price updating can render the demand information gathered by price experimentation less informative or inaccurate. We design a nonparametric learning algorithm to facilitate coordinated dynamic pricing, in which competitive firms estimate their demand functions based on observations and adjust their pricing strategies in a prescribed manner. We show that the pricing decisions, determined by estimated demand functions, converge to underlying equilibrium as time progresses. {We obtain a bound of the revenue difference that has an order of $\mathcal{O}(F^2T^{3/4})$ and a regret bound that has an order of $\mathcal{O}(F\sqrt{T})$ with respect to the number of the competitive firms~$F$ and $T$.} We also develop a modified algorithm to handle the situation where some firms may have the knowledge of the demand curve.
}%

\HISTORY{Received: February 15 2022; accepted: July 2023 by Geoffrey Parker after three revisions.}

\maketitle

%


\section{Introduction}\label{introduction}
Nowadays, online platforms/marketplaces become more centralized and offer possibly various algorithmic services to firms on them. \emph{Demand learning} can be seen as one of these services.
The demand learning algorithm is used to maximize a firm's revenue over a finite-period planning horizon, given that a firm may not know the underlying demand curve \textit{a priori}. In particular, in emerging markets, where the demand information is scarce, firms constantly adjust their pricing to collect adequate demand information. This process has been termed \textit{price experimentation}. (See \cite{broder2012dynamic} and \cite{keskin2014dynamic}.) Currently, few papers exist that consider competitive environments with unknown demand curves because such conditions add greater complexity to the analysis of the decision process. 

Evaluating the influence of decisions made by competitors is challenging. If firms are sensitive to competitors' prices, one might expect that a change in the price set by a firm would provoke an immediate response from the other firms. This situation often arises in online platform, where the selling prices are easily observed. For example, suppose that Firm A chooses to lower its price after observing a price reduction by Firm B. In such a case, neither Firm A nor Firm B is able to calculate how much of the market response is due to its price reduction and how much is the result of the competitor's drop in price as it is almost impossible for each firm to remain continuously up-to-date regarding current prices set by competitors. Pure consumer response data is difficult to obtain even if competitors' instantaneous pricing decisions remain unchanged. Such \textit{uncoordinated} price updating renders the demand information gathered by price experimentation neither informative nor accurate, and thus the demand learning procedure becomes useless. Therefore, it results in a real difficulty for firms to make better pricing decisions to extract more revenue when the competition effect is significant. 
To overcome this, firms need to adjust their prices separately over time so they can distinguish between the market response from their own pricing decisions and those from their competitors. 
The platform also wants the market to be stable. A mechanism of \emph{coordinated dynamic pricing} thus arises on the platform. Such a mechanism ensures that the pricing strategy of each firm is adjusted in a prescribed way, and the collected information is valid to learn the underlying demand curve and to make pricing decisions.

Consider a total of $F$ firms in an oligopolistic market, in which the true underlying demand curve and the presence of demand shocks are unknown. In this study, we consider the underlying demand curve for each firm as being stationary over a time horizon of $T$ periods. The firms make 
pricing decisions in each period $t=1,\cdots,T$. Firm~$i$'s decision at period $t$ is denoted by $p_t^i$, and the demand for firm~$i$'s product in period~$t$ is denoted by $D^i_t(\mathbf{p}_t)=\lambda^i(\mathbf{p}_t)+\varepsilon_t^i$, where $\mathbf{p}_t\equiv [p_t^1,\cdots,p_t^F]$ represents the pricing decisions of all firms. 
Without loss of generality, we assume that the mean value of $D_t^i$ is given by the average rate $\lambda^i$, conditional on the price vector of all companies, while $\varepsilon^i_t$ measures the demand shock in period $t$. Note that $\lambda^i(\mathbf{p})$ reflects the fact that firm $i$ acknowledges that the price decisions made by other firms will influence the demand for the product of firm $i$. By focusing on competition among the $F$ firms, we do not consider capacity limitation, production cost or marginal cost. 
We consider that the competition between firms as a noncooperative game, where firms are unable to collude. The goal of each firm is to sequentially set a price to maximize its own revenue under demand uncertainty and competition, and the pricing decisions can achieve a Nash equilibrium. Note that the price set by a firm is the only information that its competitors can observe. Because the sales quantity is usually private, firms are unable to estimate the demands of the competitors in our setting. To achieve the equilibrium pricing, firms can only submit the estimated demand function to an online platform/marketplace where this information can further be used to produce and publish the equilibrium prices.


This paper designs an algorithm of coordinated dynamic pricing in which firms in competition estimate their demand functions based on observations and adjust their pricing strategies in a prescribed manner to maximize their revenues. 
The goal of the process is for the pricing decisions, determined by the estimated demand functions, to converge to underlying equilibrium decisions. The main question that researchers aim to answer is whether such a mechanism of periodically coordinated price updating is approximately optimal for all firms. Specifically, we ask whether the mechanism may allow prices to reach a stable state and how much regret firms incur by employing such a dynamic equilibrium pricing algorithm. 


\subsection{Our Results and Main Contributions}\label{ourresults}
We propose a competitive demand learning (CDL) algorithm to solve the dynamic pricing decisions of each firm in competition. As the true demand curve is unknown, each firm estimates a demand curve via a linear approximation. Firms submit their estimated demands to the platform where an equilibrium pricing which will be later used to collect new demand information can be produced. The process is termed \emph{learning and earning} in the literature. Throughout this paper, we use period~$t$ or, equivalently, time~$t$. Let $p^{i}_t$ denote the price of firm $i$ at time $t$, and let $p^{-i}_t$ denote the prices of the other competitors. A firm is called a \emph{clairvoyant} firm if the firm has knowledge of the underlying demand curve and the distribution of demand shocks. If all firms are  clairvoyant, they will set price at the clairvoyant Nash equilibrium $\mathbf{p}^\ast$, where $\mathbf{p}^\ast$ denotes the vector consisting of all firms' prices. The goal of CDL is to make $p^i_t$ converge to the clairvoyant equilibrium price of firm $i$, $p^{i\ast}$, as $t$ grows large. The learning process is called \textit{complete} learning if the collected demand information is enough to make the estimated parameters converge to the true values; otherwise, it is termed \textit{incomplete} learning.
In a competitive environment, it is hard for a firm to achieve complete learning alone due to the influence of strategies employed by competitors. 

We summarize our contributions in the following.
\begin{itemize}
    \item \textbf{Convergence of prices to the Nash equilibrium.}  We establish the convergence results in Section~\ref{analysis}. 
   We derive the convergence of prices from the property of a contraction mapping and show the uniqueness with high probability of the price generated by the algorithm in each period. Our key result (Theorem~\ref{theorem1}) states that the best response function derived by CDL with a quadratic concave revenue function\footnote{A multiplication of a linear approximation of demand and price results in a quadratic concave function in price.} will generate the sequence $ \lbrace \mathbf {p}_t \rbrace $ that converges to $ \mathbf {p}^\ast $ as $t$ grows large. Hence, CDL ensures all firms in the market who coordinately adjust their prices achieve complete learning.  Moreover, in an alternative scenario where some firms are clairvoyant and thus unwilling to participate in price experimentation, we propose a modified CDL algorithm to account for this. In the modified CDL algorithm, using a linear approximation of the demand function that does not estimate the effects of firms with known demand is similar to that presented in \cite{cooper2015learning}. Our result (Theorem~\ref{theorem4}) shows that the sequence $\lbrace \mathbf{p}_t \rbrace$ generated by the modified CDL algorithm can still converge to $\mathbf{p}^\ast$. 
    \item \textbf{Regret bounds in a competitive environment.} In our paper, a firm's regret is described by the difference between the realized revenue and optimal expected revenue, given that other competitors adhere to the pricing decisions determined by the pricing algorithm where the optimal
expected revenue may be obtained when the firm has full demand information. We derive an $\mathcal{O}(F\sqrt{T})$ regret upper bound for the CDL algorithm (Theorem~\ref{theorem3}). In our setting, the definition of regret is distinct from that which is commonly used in a monopoly setting in which a firm plays against the market demand \emph{without} competitors. The clairvoyant problem in a monopoly setting is defined for a firm with full information about the underlying demand curve to maximize the total revenue over the planning horizon (see \cite{keskin2014dynamic}). Our definition of regret is also different from what is used in a cooperative pricing setting in which the clairvoyant problem is defined for each firm with full information about the underlying demand curve to cooperatively decide each firm's optimal price, leading to the highest total revenue (see \cite{meylahn2022learning}). In addition, we analyze the revenue difference obtained by the algorithm from the revenue obtained by the clairvoyant Nash equilibrium $\mathbf{p}^\ast$, and show that CDL ensures a bound of $\mathcal{O}(F^2 T^{3/4})$ on the revenue difference (Theorem~\ref{theorem2}).  
\end{itemize}

\subsection{Related Work}\label{relatedwork}
To conclude this section, we briefly review the literature that is most related to our work.
\begin{itemize}
\item \textbf{Learning algorithms for pricing models in a monopoly market.} Dynamic pricing is a critical tool for revenue management. The problem of unknown demand functions has been widely explored in the literature. One common theme of the existing work focuses on the task firms face to arrive at optimal pricing decisions when the underlying demand curve is unknown and how to lower the growth rate of regret.  The structure relies on price experimentation to infer accurate information about the demand function. Since price experimentation can be costly, some papers focus on a balance between the exploration-exploitation trade-off. \cite{broder2012dynamic} used maximum-likelihood estimation to build the pricing policy and \cite{den2014simultaneously} proposed a policy of gradually reducing the taboo interval to achieve asymptotically optimal pricing. The two policies that they proposed can be regarded as iterated least squares. \cite{keskin2014dynamic} provided a generally sufficient condition for iterated least squares methods which ensures that regret does not exceed the order of $\sqrt{T}$. \cite{ban2021personalized}  formulated the feature-based demand model and designed a dynamic pricing policy where the regret is related to the dimension of the feature vector. 
A common characteristic of previous work is that the underlying demand model is assumed to be the same as the parametric demand function and is thus well-specified. \cite{besbes2015surprising} proposed a pricing policy constructed on the assumption that the underlying demand curve is non-linear but used a linear approximation to model the market response to the offered price. We refer readers to \cite{chen2019coordinating} which extend the model to coordinate pricing and an inventory control problem.

\item \textbf{Dynamic pricing in a competitive environment.}  Few papers consider dynamic pricing with competition, as this type of problem requires a game-theoretic approach to obtain a solution. Some studies consider non-cooperative competition in revenue management, where the products are differentiated, for example, at the service level or based on certain attributes of the goods on which consumer choice closely depends. \cite{levin2009dynamic} presented a dynamic pricing model in which strategic consumers choose differentiated perishable goods. With firms able to benefit from a lack of complete information among consumers, they provide an analysis of the equilibrium price dynamic under different market settings. \cite{gallego2014dynamic} also studied how consumer choice depends not only on price, but also on purchase time and product attributes; they showed that the shadow price solved by the deterministic problem can be used to construct an asymptotic equilibrium pricing policy. When different companies sell identical goods, the market response depends mainly on the price. \cite{cooper2015learning} used a ``flawed" demand function without incorporating competitor prices and showed that (a) the prices converge to the Nash equilibrium if the slope is known to the seller, (b) the prices converge to the cooperative prices if the intercept is known to the seller, and (c) there are many potential limit prices that are neither Nash equilibrium nor cooperative prices if the parameters of the demand function are unknown.

\item \textbf{Learning algorithms for pricing models in a competitive environment.} 
If the content of learning is incorporated into the pricing policy, firms are able to learn the response of consumers and competitors, continuously adjusting pricing to achieve optimality over time.
\cite{bertsimas2006dynamic} considered a myopic pricing policy in a competitive oligopolistic environment. In such a case, a given firm must estimate not only its own demand but also the demand and pricing of competitors. A method of mathematical programming with equilibrium constraints is applied to formulate the model to estimate the competitor parameters.  \cite{kachani2007modeling} proposed further improvements that adjusted the most accurate parameters of demand function periodically based on the equilibrium demands. \cite{gallego2012demand} considered a firm selling a single product with multiple versions and developed a multinomial logit choice demand model, commonly used in the context of non-cooperative competition between firms selling differentiated goods, each seeking to maximize total revenue. Bayesian updating is used to estimate the unknown arrival rate, and maximum likelihood estimation is used to update the core value. They showed that the unknown parameters can be estimated simultaneously while sales are changing. \cite{kwon2009non} used the observed sales data to estimate the parameters of demand by Kalman filtering, showing that a differential variational inequality can be used to model the non-cooperative competition.
\cite{fisher2018competition} conducted two field experiments with randomized prices. The first experiment estimated a consumer choice model in which observations of sales made by competitors are not required, and the second tested a best response pricing strategy. 

More recently,  \cite{meylahn2022learning} considered a pricing problem in a  duopoly  and  proposed a learning algorithm based on which firms learn to collude instead of compete. They showed that prices will converge to maximize the earnings of the firm if it is profitable for the colluding firms, in which case, it is joint revenue; otherwise, the convergence will reflect a Nash equilibrium.  \cite{golrezaei2020no} studied a pricing problem under reference effects in a duopoly, considering the reference effects as a third firm and applying an online mirror descent algorithm to determine the optimal pricing. Notably, the authors assume that the firm has full information about the gradient of the revenue function, but in our problem, the demand function needs to be learned.
\end{itemize}

\section{Model and Preliminaries}\label{model}
We consider a periodical equilibrium pricing problem for $F$ firms over a planning horizon of $T$ periods. The products offered by firms are substitutes. In each period $t=1,\cdots,T$, each firm needs to set prices $p_t^i$, chosen from a feasible and bounded policy set $\mathcal{P}^i=\left[p^{i,\ell},p^{i,h}\right]$, $p^{i,\ell}<p^{i,h}$, $\forall i=1, \cdots, F$. The prices set by firms affect the market response of all firms in competition. Recall that $\mathbf{p}\equiv (p^i,p^{-i})$ denotes the vector of prices of all firms in the competition. The market response to the price $p_t^i$ for firm $i$ at time $t$ (which is exactly the demand function) is given by $D_t^i(\mathbf{p}_t)=\lambda^i(\mathbf{p}_t) +\varepsilon_t^i,\; \forall i =1,\cdots,F,$
in which $\lambda^i(\mathbf{p}_t)$ is a deterministic twice differentiable function representing the mean demand during a period for every period, conditional on the price $\mathbf{p}_t$, and $\varepsilon_t^i$ are zero-mean random variables, assumed to be independent and identically distributed. Hence, the underlying demand curve $\lambda^i(\mathbf{p})$ of firm $i$ not only depends on the price $p^i$, chosen by itself, but also on the prices of other firms $p^{-i}$, where $p^{-i}=\left\lbrace p^1,\cdots,p^{i-1},p^{i+1},\cdots,p^F \right\rbrace$. The revenue function $r^i$ of firm $i$ obtained from prices $\mathbf{p}$ is denoted by $r^i(\mathbf{p}) \equiv p^i \mathbb{E}[D^i(\mathbf{p})]$. Each firm seeks to maximize its revenue in a competitive environment.  
Thus, firms want their pricing decisions achieve a pure Nash equilibrium (NE). We make the following assumptions, which are necessary for the results obtained in this paper, about the mean demand functions and revenue functions:
\begin{assumption} \label{assumtion1}
(i) For any $p^i\in\mathcal{P}^i$, $\displaystyle{\frac{\partial \lambda^i(\cdot,p^{-i})}{\partial p^i}<0,\forall i=1,\cdots,F}$.\\
(ii) For any $p^i\in\mathcal{P}^i$, $\displaystyle{\frac{\partial \lambda^i(p^i,p^{-i \setminus j},\cdot)}{\partial p^j}>0,\forall j \neq i,i=1,\cdots,F}$.\\
(iii) For every $r^i(\mathbf{p})$, $$\sum_{j\neq i}^F \left\vert \frac{\partial^2 r^i}{\partial p^i \partial p^j}\right\vert < \left\vert \frac{\partial^2 r^i}{\partial p^{i2} }\right\vert, \forall p^i \in \mathcal{P}^i, p^j \in \mathcal{P}^j.$$\\
(iv)  For every firm $i$, there exists a constant $s_0$ such that, for all $\vert s\vert< s_0$, $\mathbb{E}\left[\exp\left\lbrace s\varepsilon_1^i\right\rbrace\right]<\infty$, and the variance of each firm's $\varepsilon^i$ is equal to $\sigma^2$.\\
(v) For every firm $i$, given $p^{-i}$, firm $i$ chooses to price at a price $p^i\in\mathcal{P}^i$ satisfying $\mathbb{E}\left[D^i(p^i,p^{-i})\right]\geq 0$.
\end{assumption}

Assumption~\ref{assumtion1}. (i) ensures that for every firm $i$, the underlying demand function $\lambda^i (\cdot,p^{-i})$ is strictly decreasing on $p^i$ given the prices set by other firms, $p^{-i}$. Assumption~\ref{assumtion1}. (ii) dictates that $\lambda^i(p^i,p^{-i \setminus j},\cdot)$ is strictly increasing on $p^j$ with $p^i$ and $p^{-i \setminus j}$ given, in which $p^{-i \setminus j}$ represents the vector constituted by all prices except $p^i$ and $p^j$. 
 Assumption~\ref{assumtion1}. (iii) is termed as the ``diagonal dominance'' condition.  Assumption~\ref{assumtion1}. (iv) ensures that the demand shock $\varepsilon^i_t$ of each firm has a light-tailed distribution and the homogeneity of variance.\footnote{The assumption of homogeneity of variance is to ease our analysis. Note that the non-homogeneity of variances would make the problem more realistic. This setting would lead to a slight change in the analysis. The constants of the upper bound of the results in all the propositions and theorems are related to the variance. Thus, with the non-homogeneity of variances, these constants should be chosen according to the largest variance of the demand shocks among these firms.}

We give some examples of demand functions from \cite{bernstein2004dynamic} and \cite{chen2011group} that satisfy Assumption~\ref{assumtion1}. (i) and (ii). 
\begin{example}
Linear demand: $\lambda^i(\mathbf{p})=\alpha^i-\beta^{ii}p^i+\sum\limits_{j=1,j\neq i}^F\beta^{ij}p^j$, $\alpha^i>0$, $\beta^{ii}>0$, {$\beta^{ij}>0$}.
\end{example}
\begin{example}
Multinomial logit demand: $\lambda^i(\mathbf{p})=\frac{\exp^{\alpha^i-\beta^ip^i}}{1+\sum\limits_{i=1}^F \exp^{\alpha^i-\beta^ip^i}}$, $\alpha^i>0,\beta^i>0$ and $\alpha^i-\beta^ip^i<0$ for $p^i \in \mathcal{P}^i$.
\end{example}
\begin{example}
Exponential demand: 
$\lambda^i(\mathbf{p})=\exp^{\alpha^i-\beta^{ii}p^i+\sum\limits_{j=1,j\neq i}^F\beta^{ij}p^j}$, $\alpha^i>0$, $\beta^{ii}>0$, {$\beta^{ij}>0$}.
\end{example}
\begin{example}
Semi-log demand: $\lambda^i(\mathbf{p})=\alpha^i-\beta^{ii}\log p^i+\sum\limits_{j=1,j\neq i}^F\beta^{ij}\log p^j$, $\alpha^i>0$, $\beta^{ii}>0$, {$\beta^{ij}>0$}.
\end{example}
The above parameter restrictions ensure the existence of equilibrium, and along with Assumption~1. (iii), i.e., the diagonal dominance condition, furthermore guarantees the uniqueness of equilibrium, e.g., $\beta^{ii}>\sum_{j,j\neq i}^F\beta^{ij}$ for the linear demand model (See \cite{milgrom1990rationalizability}). Assumption~1. (iii) can be assumed equivalently as $\sum\limits_{j=1}^F\frac{\partial \lambda^i(\mathbf{p})}{\partial p^j} <0$ for all $i=1,\cdots,F$. More discussions are seen in \cite{corchon1996stability,bernstein2004general}, and \cite{tuinstra2004price}.
\footnote{Note that the diagonal dominance condition ensures that the Hessian matrix of the revenue function is negative definite, but the negative definiteness of the Hessian matrix of the revenue function does not ensure the diagonal dominance.}


For every firm $i$, if the underlying demand curve $\lambda^i(\cdot)$ and the distribution of the demand shock $\varepsilon^i$ are known \textit{a priori}, 
and given the competitors' pricing at $p_t^{-i}$ for each period $t$, firm $i$ wishes to solve the following optimization problem
\begin{equation}
  \begin{array}{ccl}
 &\mbox{maximize}_{p^i} & \sum\limits_{t=1}^Tr^i(p^i_t,p_t^{-i})\\[6pt]
    &\mbox{subject to}&  p^i_t \in \mathcal{P}^i, \forall i=1,\cdots,F,
  \end{array} \label{singleagent}
\end{equation}
where $r^i(p^i_t,p_t^{-i})\equiv  p^i \mathbb{E}\left[D^i_t(p^i,p_t^{-i})\right]$. The optimal pricing at each period could be solved through the best response function, that is
$$z^i(p^i,p_t^{-i})=\argmax_{p^i\in \mathcal{P}^i} r^i(p^i,p_t^{-i}).$$
A feasible vector $\mathbf{p}^\ast$ is a solution to the problem of NE if, for all firms $i=1,\cdots,F$, we have
$$\displaystyle{r^i(p^{i\ast},p^{-i\ast}) \geq r^i(p^{i},p^{-i\ast}) \quad \forall p^i \in \mathcal{P}^i \geq 0}.$$
Suppose all firms have knowledge of their own underlying demand curve  $\lambda^i(\cdot)$ and distribution of the demand shock $\varepsilon^i$, then all firms wish to solve (\ref{singleagent}). In this case, the problem becomes a classical noncooperative game. 
It is well known that if there exists a vector $\mathbf{p}^\ast$ and multipliers that satisfy the following centralized KKT system,

\begin{equation}
\begin{array}{cl}
\begin{array}{cc}
\lambda^i(\mathbf{p})+p^i \nabla_{p^i} \lambda^i(\mathbf{p})+\mathbb{E}(\varepsilon^i)
\end{array}
-\mu^{i,1}+\mu^{i,2}=0 & \forall i\\[6pt]
\mu^{i,1}\geq 0, \mu^{i,1}\cdot (p^{i,h} -p^i)= 0,p^{i,h} -p^i\geq 0 &\forall i\\[6pt]
 \mu^{i,2}\geq 0, \mu^{i,2}\cdot(p^i-p^{i,l}) = 0, p^i-p^{i,l}\geq 0,&\forall i,
\end{array}\label{KKT}
\end{equation} then the vector $\mathbf{p}^\ast$ is a NE. The single-period equilibrium pricing problem is reformulated as the KKT system (\ref{KKT}). Note that the KKT system (\ref{KKT}) can also be replaced by projecting the solution of the Lagrange function into the feasible space.
Therefore, if the underlying demand curve and distribution of the demand shock are known in advance, the prices of all firms reach NE $\mathbf{p}^\ast$ at the beginning of the planning horizon, and the pricing decisions of all firms will not change at any period $t$, that is, $p^i_1=p^i_2=\cdots=p_T^i=p^{i\ast}$.
\subsection{Design Ideas of the Algorithm} Owing to a lack of information about the underlying demand curve, firms must observe the sales quantity to estimate the demand function. A crucial question is ``how to use the collected demand information to help firms learn their demands?'' First, firms need to distinguish the demand change caused by its own price adjustment and that caused by competitors' price adjustment. If a firm lowers its price while a competitor increase its price, the collected sales increment is insufficient to reflect the firm's demand-price relationship.\footnote{This relationship is often estimated as a ``gradient'' of the underlying demand curve. Unlike other traditional gradient-based approaches, we do \emph{not} assume that the gradient or other information of the objective function is available. Thus, it is a more challenging environment for the analysis.} Therefore, our proposed algorithm requires firms to separately adjust their pricing decisions in a coordinated manner. The intuition behind this approach is described in Lemma~\ref{lemma1}.

A pricing algorithm should create suitable price dispersion to guarantee that the parameter estimates converge to the true values such that the prices eventually converge to the clairvoyant $\mathbf{p}^\ast$. A learning process relying on these phenomena is called \textit{complete learning} in literature. On the contrary, \textit{incomplete learning} describes a phenomenon in which the collected demand information is too scarce to be used to update the estimates, causing the firm to use a price that is not optimal. As pointed out by \cite{keskin2014dynamic}, which studied the dynamic pricing problem in the monopoly setting, using myopic policy will lead to incomplete learning and, thus, poor profit performance. Our proposed algorithm can help each firm achieve complete learning when there are competitors and minimize revenue loss by making good pricing decisions.

In practice, firms usually cannot observe the demand information of the competitors directly, and the online platform owns this information. Instead of the firm estimating the competitors’ demand functions, it could instead submit the estimates of their demand curves to the online platform. Then, the online platform performs the calculation and informs the participating firms of an equilibrium price based on all firms' estimates. The information-sharing structure of our paper is slightly different from that of  \cite{liu2021information}, which studied the information-sharing problem between the online platform and retailers. \cite{liu2021information} assumed that the online platform owns the information and makes information-sharing decisions rather than the sellers. In this case, the objective of the online platform is to maximize the sellers' total profits; however, the online platform charges a fixed percentage of the sellers' total profits as a commission fee. They found that an incentive exists for the online platform to share information with the sellers. In comparison, we aim to design an algorithm that can help firms without the knowledge of the underlying demand curve make pricing decisions with a sublinear regret, noting that we do not include a commission fee charged by the online platform.



\subsection{Algorithm}\label{algorithm}
We propose a CDL algorithm to solve the pricing decisions of each firm in competition. CDL operates in \textit{stages}, which we index by $n$ for $n=0,1,2,\cdots$. Every stage is equally separated as $F+1$ time intervals, which we index by $m$ for $m=1,2,\cdots,F+1$, and each interval contains $I_n$ periods (i.e., time steps). 
The length of intervals $I_n$ is exponentially increasing as $n$ increases. In the beginning of stage $n$, firm $i$ publishes a price $\hat{p}_n^i$, the vector of all prices is denoted as $\hat{\mathbf{p}}_n$. As there are $(F+1)I_n$ periods in one stage, the notation $t_{n}$ represents period $t$ at the beginning of stage $n$. In Step 1, firm $i$'s price $p_t^i=\hat{p}_n^i$ will not be changed in $F I_n$ but will be changed in one time interval if the index $m$ is equal to the firm index $i$. For example, firm 1 needs to set  $p_t^1=\hat{p}_n^1+\delta_n$ for $t=t_n+I_n+1,\cdots,t_n+2I_n$ and set $p_t^1=\hat{p}_n^1$ for the remaining periods of stage $n$.
Therefore, all firms are required to adjust their prices sequentially and collect sales information. Using the observations in stage $n$, the demand functions $\lambda^i (\cdot)$ of all firms during a period are then approximated by a simple linear function. We separate one stage into $F+1$ parts because each firm must estimate the $F+1$ parameters of the linear function. For index $i=1$, for example, firm $1$ can use the collected sales information from the interval $t_n+1,\cdots,t_n+I_n$ and the interval $t_n+I_n+1,\cdots,t_n+2I_n$ to estimate $\beta^{11}$ through the finite difference method.
At the end of stage $n$, all firms submit their approximation demand functions to the online platform,  and the online platform uses these to find the current NE~$\mathbf{\hat{p}}_{n+1}$.
Now, we present in detail the CDL algorithm looping $n$ from $0$ until a terminal stage, given as period $T$.
\begin{itemize}
\item Step 0. Preparation:
If $n=0$, input $I_0$, $v$, $t_0=0$ and $\hat{p}_1^i,\forall i=1,\cdots,F$. If $n>0$, set $I_n=\lfloor I_0v^n \rfloor$ and $\delta_n=I_n^{-\frac{1}{4}}$. 
\item  Step 1: Setting prices. The rule of firm $i$'s price $p_t^i$ at time $t$ is 
\begin{equation*}
    \begin{array}{ll}
     p_t^i=\hat{p}_n^i,    & \forall t=t_{n}+1,\cdots,t_{n}+iI_n,t_{n}+(i+1)I_n+1,\cdots,t_{n}+(F+1)I_n,  \\
      p_t^i=\hat{p}_n^i+\delta_n, & \forall t=t_{n}+i I_n+1,\cdots,t_{n}+(i+1)I_n.
    \end{array}
\end{equation*}
 Set $t_{n+1}=t_{n}+(F+1)I_n$.
\item Step 2. Estimating:
$$\displaystyle{(\hat{\alpha}^i_{n+1},\hat{\beta}^{ij}_{n+1})=\arg \min_{\alpha^i,\beta^{ij}} \left\lbrace \sum_{t=t_n+1}^{t_n+(F+1)I_n} \left[D^i_t-\bigg(\alpha^i-\beta^{ii} p_t^i +\sum_{j=1,j\neq i}^F \beta^{ij} p^j_t\bigg)\right]^2 \right\rbrace.}$$
\item Step 3. Computing the equilibrium:
We define the following optimization problem for firm $i$:
$$\displaystyle{\max_{p^i} r^i_{n+1} \equiv \max_{p^i} G_{n+1}\left\lbrace p^i,p^{-i},\hat{\alpha}^i_{n+1},\hat{\beta}^{ij}_{n+1} \right\rbrace},$$
where $G_{n+1}\left\lbrace p^i,p^{-i},\hat{\alpha}^i_{n+1},\hat{\beta}^i_{n+1} \right\rbrace $
$$\displaystyle{\equiv \left\{ p^i \left(\hat{\alpha}^i_{n+1}-\hat{\beta}^{ii}_{n+1}p^i+\sum\limits_{j=1,j\neq i}^F \hat{\beta}^{ij}_{n+1}p^{j}\right) \middle| \hat{\alpha}^i_{n+1}-\hat{\beta}^{ii}_{n+1}p^i+\sum\limits_{j=1,j\neq i}^F \hat{\beta}^{ij}_{n+1}p^{j} \geq 0, p^i \in \mathcal{P}^i\right\}.}$$
Proceeding to solve the system: 

\begin{equation}
\begin{array}{cl}
\begin{array}{cc}
\hat{\alpha}^i_{n+1}-2\hat{\beta}^{ii}_{n+1}p^i+\sum\limits_{j,j\neq i}^F \hat{\beta}^{ij}_{n+1}p^{j} 
\end{array}
-\mu^{i,1}+\mu^{i,2}=0 & \forall i,\\[6pt]
\mu^{i,1}\geq 0, \mu^{i,1}\cdot\left(p^i-p^{i,h}\right)=0,p^{i,h}-p^i \geq 0 &\forall i,\\[6pt]
\mu^{i,2}\geq 0, \mu^{i,2}\cdot\left(p^{i,l}-p^i\right) = 0, p^i-p^{i,l}\geq 0 &\forall i.
\end{array}\label{DDEPKKT}
\end{equation}
Then, prices for each firm $\hat{p}^i_{n+1}$ are set to the solution of this system. The uniqueness of the solution is discussed in Lemma~\ref{lemma2}. If there exist multiple solutions, pick any one. Set $n=n+1$ and \textbf{return to Step 0}.
\end{itemize}

During the CDL operation cycle in Step 1 to explore the sensitivity of the market to price changes, each firm must adjust the price by adding $\delta_n$. In Step 2, the underlying demand function of all firms is estimated using linear regression. Note that we only use the data gathered in this stage since past data may lead to inaccurate estimation of the demand function. If the underlying demand curve is a linear function, we can use all the cumulative data for the estimation due to the unchanging market sensitivity provided that price changes. However, as the underlying demand curve is not known \textit{a priori}, market sensitivity may keep changing as the prices change if the demand function is misspecified. 
In Step~3, the mechanism solves the KKT system  (\ref{DDEPKKT}), which includes the revenue optimization problem of all firms, providing a way to compute the equilibrium under the approximated demand. The solutions to system (\ref{DDEPKKT}) will be used for the next stage. Because the demand experienced by each firm is private data, firms are unable to estimate each other's demand levels, firms will submit their estimates to the online platform which computes the equilibrium prices for the next stage. In practice, an online platform notifies each firm in advance of the precise price adjustment at the start of each stage. Subsequently, each firm sets its prices in advance based on the notification, guaranteeing that prices are promptly updated in accordance with the rules of the price experiment. During the price experiment, each firm monitors its own demand and the prices of other firms, utilizing this data to estimate their respective demand functions (learning their demand functions) at the end of the stage. The online platform should offer guidance to firms regarding this mechanism, ensuring their comprehensive understanding of its functioning and significance. The platform can enforce compliance with this mechanism through contractual agreements and may impose penalties on firms that violate it. For instance, E-commerce platforms like Amazon are highly dynamic, enabling sellers to continually monitor the prevailing market situations and to adjust their prices accordingly in real time. These platforms can motivate participating sellers to adhere to this mechanism by offering incentives such as enhanced visibility or reduced transaction fees. Similarly, a hotel booking platform can involve all participating hotels in experiments and encourage them to update their room rates based on this mechanism. The platform can also modify the visibility of hotels on its platform in response to deviations from the prescribed updating rules.

\subsubsection*{Incentives}
A natural question that arises is why firms are willing to abide by the algorithm. The reason is that the prices set in Step 3 for every stage are the equilibrium prices at which firms have no incentives to deviate unilaterally. If a firm wants to change its price unilaterally, the price experiment (learning) fails, and such an action results in incomplete learning, which is undesirable.  In the \emph{monopoly} setting, complete learning ensures that the proportion of the firm’s expected lost revenue in each period grows sublinearly over time, whereas incomplete learning causes that proportion to grow linearly over time. In other words, a firm achieves complete learning such that the estimated parameters converge to the true values, and its price gradually converges to the optimal price, which maximizes its revenue. As for the \emph{competitive} setting, complete learning causes the estimated parameters to converge to the true values, and thereby under non-cooperative competition, the price of each firm converges to the equilibrium price.
\begin{figure}[!ht] \label{figure1}
		\centering
   \includegraphics[scale=1.8]{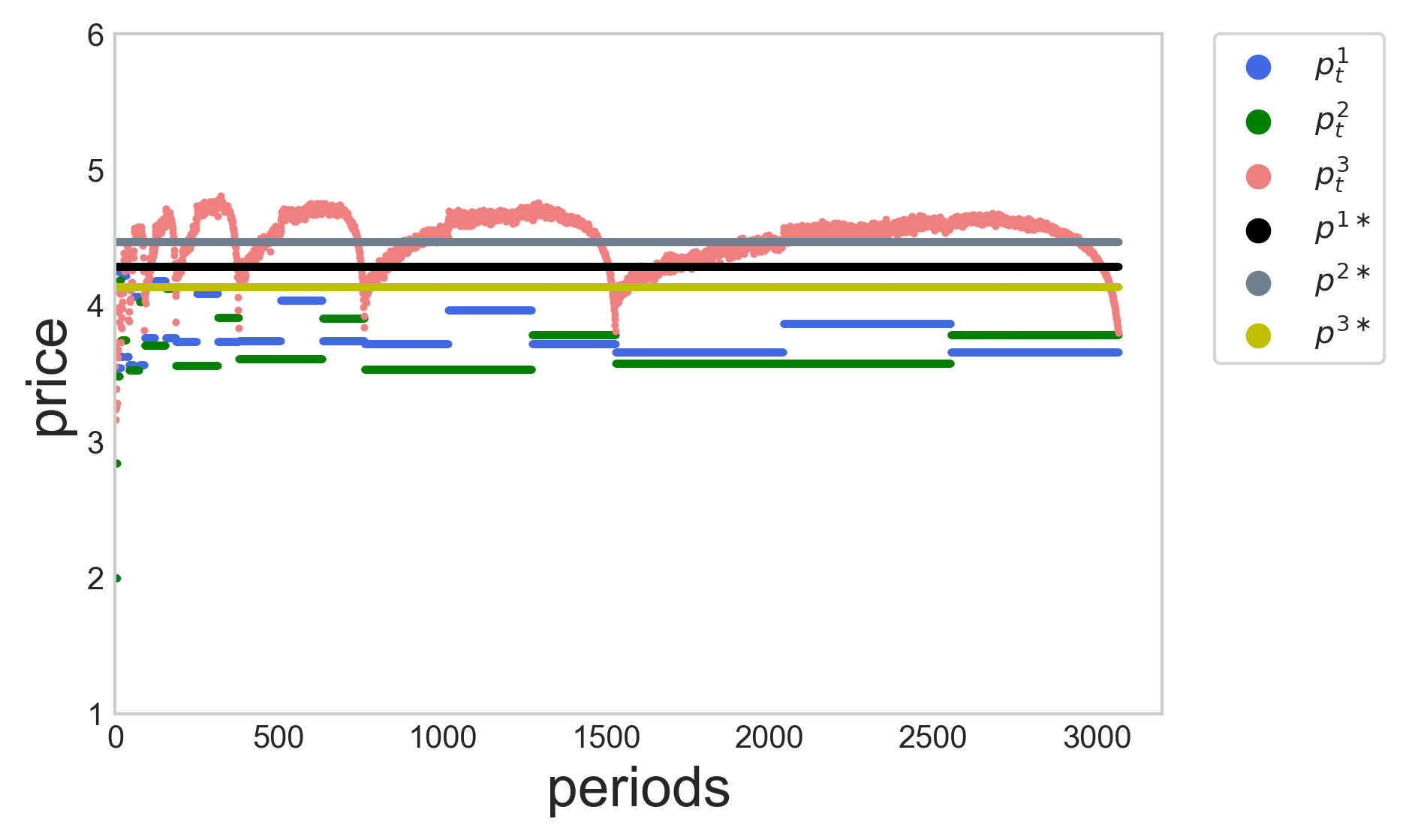}		
 	\caption{The prices sequences of firm~1, firm~2 and firm~3. }
\end{figure}

We put forward an exemplar numerical experiment to illustrate the importance of joining the price experiments and complete learning. As shown in Figure 1, the prices sequences $\left\lbrace p_t^1 \right\rbrace$ and $\left\lbrace p_t^2 \right\rbrace $ of firm~1 and firm~2 are generated by the CDL algorithm, respectively, and the prices sequence $\left\lbrace p_t^3 \right\rbrace$ of firm~3 is generated by the ``myopic" policy. According to the algorithm, firm~1 and firm~2 use a linear model $\alpha^i-\beta^{ii}p^i+\beta^{ij}p^j$, where $i$ and $j$ are $1$ and $2$, respectively, to estimate their demands. Firm~3 uses a myopic policy that combines greedy price update with sequential estimation as follows: at each period, firm~3 uses $\alpha^i-\beta^{ii}p^i+\sum_{j\neq i} \beta^{ij}p^j$ as the model to estimate firm 3's demand and then sets price to maximize its expected revenue, assuming that the unknown demand parameters are equal to the most recent estimates.

We randomly draw 100 instances for the underlying demand curves, where the parameters $\alpha^i$ are uniformly drawn from $[3,5]$, $\beta^{ii}$ are uniformly drawn from $[0.8,0.9]$ and $\beta^{ij}$ are uniformly drawn from $[0.3,0.35]$.
The standard deviation of demand shock is set at $0.16$. Each point in Figure~1 is computed by averaging $100$ repeated runs. We compute the averaged faction of revenue loss and averaged fraction of revenue difference  (which are formally defined in Section \ref{numericalexperiments}) for each firm. The averaged faction of revenue loss and revenue difference of firm 1 are around $12\%$ and $23.6\%$ respectively, of firm~2 are around $10.9\%$ and $25\%$ respectively, and of firm~3 are around $11\%$ and $31\%$ respectively. We can actually see that a firm who deviates from our proposed algorithm suffers revenue loss roughly around the others' and more revenue difference than the others' in our experimental settings.  It should be noted that the averaged fractions of revenue loss and revenue difference remain constant over time, indicating that both regret and revenue difference grow linearly over time. Moreover, it can be observed that the prices set by all firms are unable to converge to the true Nash equilibrium. Hence, there is an incentive for all firms to obey the algorithm since deviating from it will cause significant revenue loss.

\subsection{Performance Measures}\label{performance} 
We say a policy is \emph{admissible}  if it depends only on information that is available to the firm. Define the history available for firm $i$ at time $t$ as 
$\mathcal{H}_{t}^i=\left\lbrace\mathbf{p}_1,{D}^i_1(\mathbf{p}_1),\cdots,\mathbf{p}_{t-1},{D}^i_{t-1}(\mathbf{p}_{t-1})\right\rbrace.$  
The pricing decision of firm $i$ can be expressed as a policy function of the history, that is, 
$p_1^i=\pi^i(\emptyset)$ and $p_t^i=\pi^i(\mathcal{H}^i_{t})$ for all $t=2,\cdots,T$.
We denote by $\Pi=(\pi^1,\cdots,\pi^F)$ the admissible pricing policies of all firms.  
 Therefore, the pricing policy $\Pi$ assigns the prices of all firms to be used in the next period for each possible history of sales data.\footnote{It can be seen from the above definition of admissible pricing policies that the competitors' demands are not observable to the firm.}

For each feasible pricing policy $\Pi$ generated by the algorithm, we examine it in two ways. The first method is to evaluate whether $\hat{\mathbf{p}}_n$ will converge to $\mathbf{p}^\ast $; the second is to measure the impacts of $\hat{\mathbf{p}}_t$ on each firm $i$'s revenues. 
In the literature, this impact is quantified as \textit{regret}, and the cumulative regret over $T$ periods is defined as
\begin{equation}
    \displaystyle{\mathcal{R}^i(\pi^i,T)=\mathbb{E}\left[\sum\limits_{t=1}^T r^i(p_t^{i\ast},p_t^{-i})\right]-\mathbb{E}\left[\sum\limits_{t=1}^T r^i(p_t^i,p_t^{-i}) \right],}
    \label{regret}
\end{equation}
where $p_t^{i\ast}$ is computed by the best response function $z^i(p^i,p_t^{-i})$ if firm $i$ knows his underlying demand function, and $p_t^{-i}$ represents the prices of other firms at time $t$ determined by the algorithm. 
Note that $p_t^{-i}$ might be different at different time. Therefore, $p_t^{i\ast}$ will correspondingly change according to the best response function.
We are most concerned with whether the average regret per period, $\mathcal{R}^i(\pi,T)/T$ converges to $0$ with an acceptable convergence rate. 

The definition of regret in the monopoly setting is worst-case regret, which is the revenue of a single firm generated by a given pricing policy compared to the optimal revenue generated by the firm with the knowledge of the demand curve. However, this definition is not suitable for our setting since revenues obtained at clairvoyant NE $\mathbf{p}^\ast$ are not optimal. Alternative pricing vectors $\mathbf{p}$ that lead to higher revenues than those at $\mathbf{p}^\ast$ may exist when firms search for an equilibrium point (see also in \cite{cooper2015learning}). {What should be emphasized here is that this situation, in which price vectors are \emph{not} Nash equilibrium points and incentives remain for firms to change their prices, may happen when all firms engage in non-cooperative competition.}
From the perspective of the incentives for deviating from the prices set by the algorithm, a related topic is algorithmic collusion, which refers to the scenario in which firms choose to set prices that yield the highest total revenue globally. Hence, the benchmark of such a scenario is the revenue obtained by the ``optimal” price, which could be a global optimal price that maximizes the combined revenue of all firms, or an equilibrium price that maximizes each firm’s own revenue.  Thus, its regret is relative to a static optimal prices, while our definition of regret is relative to a ``dynamic" optimal price, which is the best response to the competitors’ prices. Another benchmark to consider is that, for each firm, there is an optimal price that can maximize its total revenue over the selling horizon, i.e., $p^{i\ast}_T=\argmax \sum_{t=1}^T r^i(p^{i\ast},p_t^i).$ The regret in this case is bounded by the regret as defined in our study because the revenue generated by $p_t^{i\ast}$ is the optimal for each period. 

Compared to the revenues obtained by the pricing policies $\Pi$, the revenues obtained by the clairvoyant Nash equilibrium may be higher or lower. 
We define the \textit{revenue difference} as 
\begin{equation}
    \mathcal{D}^i\left(\pi^i,T\right)=\Bigg\vert\mathbb{E}\left[\sum_{t=1}^T r^i(p^{i\ast},p^{-i\ast})\right]-\mathbb{E}\left[\sum_{t=1}^T r^i(p_t^i,p_t^{-i}) \right]\Bigg\vert
    =\Bigg\vert p^{i\ast}\lambda^i(\mathbf{p}^\ast)T-\mathbb{E}\left[\sum_{t=1}^T p_t^i D_t\right]\Bigg\vert \label{difference}
\end{equation}
to evaluate it. 
The revenue difference is positive in each period. As it is impossible to determine whether revenues obtained by the clairvoyant Nash equilibrium are less than those obtained by the pricing policies $\Pi$, the difference per period between the two is expressed as an absolute value no less than zero.

\subsection{Discussions: Our and Related Models and Algorithms}
\subsubsection*{Connections and Comparisons to Closely Related Models and Algorithms}
The structure of operation cycles has frequently been utilized prevalently in research on learning and pricing problems.  \cite{besbes2015surprising} proposed a  pricing policy to ascertain the demand function and determine the optimal pricing that corresponds to the approximate demand function.  CDL is a generalization of their study, in which there are no competitors ($F=1$), so the number of periods needed for each stage is reduced to equal to the number in their study. A later study by \cite{chen2019coordinating} considered joint pricing and inventory control with a backlog, assuming that the underlying demand model is multiplicative. \cite{golrezaei2020no} viewed the linear demand model as a first-order approximation to a complex demand model and assumed that a firm can access the gradient of the revenue function at the firm’s own price, which they termed  ``first order oracle", i.e., $\partial r^i(\mathbf{p})/\partial p^i$. In this setting, firms have no need to learn the demands and can make pricing decisions based on the first-order oracle. Nevertheless, we focus on the problem that arises in cases in which the firms have no prior information about the underlying demands, especially in an emerging market.

Some recent studies also deal with equilibrium pricing algorithms under problem settings similar to ours.   \cite{hansen2021frontiers} used an upper confidence bound (UCB) algorithm to help firms make pricing decisions that were not necessarily Nash equilibria, since each price (arm) of the UCB may yield a better payoff than that under the equilibrium. In this case, firms collude to realize the highest revenues.  Another collusion setting was discussed in \cite{meylahn2022learning} based on a two-phased price experiment covering a collusive and a competitive phase. In each of these phases, they computed two price vectors, the collusive price vector and the competitive price vector, with the firms using the price vector that yielded the highest profit for each firm.  Their definition of regret was the extent of the revenue loss between the revenue expected under generated pricing decision  $\mathbf{p}_t$ and the revenue realized under the global optimal price $\mathbf{p}^\ast$.

{\subsubsection*{Our Number of Periods for Competitive Learning: a Balance between Learning and Earning}
``How long should the firms learn?'' is a common question in this kind of problem. We define the \textit{earning phase} for a firm in the competitive environment as the periods when it sets price at an equilibrium using the currently available information from estimation and the \textit{learning phase} for a firm as the periods when it experiments its price deviating from the equilibrium price. 
In stage $n$ of the designed CDL algorithm, each firm spends $1$ interval (or $I_n$ periods) to learn; during the other $F$ intervals (or $F\times I_n$ periods), each firm earns while observing and collecting the demand information. The CDL algorithm requires one firm to set its price to $\hat{p}^i_n+\delta_n$ in its learning phase and other firms to maintain its price at $\hat{p}^i_n$, and firms that have gone through the learning phase keep the price at $\hat{p}^i_n$ for the remaining earning phase. Although each firm has to deviate from $\hat{p}^i_n$ to $\hat{p}^i_n+\delta_n$, the slight change is useful and is diminishing as $n$ increases. The purpose of this change is that these firms need to extract information about consumers' response given its selling price and the competitors' prices.

In the earning phase, each firm sees $\hat{p}^i_n$ as the current optimal decision that can induce the most revenue using the information collected so far. Each firm chooses to lose some revenue in exchange for more accurate demand information in the learning phase. The learning phase duration is still $1$ interval per stage in a competitive environment. The number of intervals spent on earning per stage is equal to the number of firms without knowledge of the underlying demand curve. During the finite number of periods $T$, the more competitors a firm has in the market, the less time this firm will spend on learning. 
The revenue losses occur in both the learning and the earning phases unless all of the firm's pricing decisions converge to the clairvoyant NE $\mathbf{p}^\ast$. 
Several other learning algorithms, such as \cite{besbes2015surprising}, \cite{chen2019coordinating}, and \cite{ban2021personalized}, are also designed to achieve the right balance between learning and earning, literally meaning that the time-averaged regrets are all shown to converge to $0$ when time progresses.

\section{Analyses: Convergence, Revenue Difference, and Regret}\label{analysis}
The goal of this paper is to investigate whether such a linear model remains suitable in a competitive setting. To this end, we start by analysing the limit point of $\{\hat{\mathbf{p}}_n\}$ given its existence. 
We divide the analyses into two parts: first, we argue by intuition that when the limit price exists and the noise is absent, no unilateral changes on limit price $\tilde{p}^i$ will maximize the revenue function for an individual player $i$; that is, limit price is exactly the clairvoyant NE. Second, we argue that this statement remains true when noise is present.

We formally show in Section~\ref{convergence} that the pricing decision of each firm will converge to the clairvoyant NE $\mathbf{p}^\ast$ in probability if all firms accept the coordinated control. We find that the average regret defined in (\ref{regret}) converges to zero and the convergence of difference defined in (\ref{difference}) is dependent on the number of competitors. We discuss the revenue difference bound and a regret bound in Section~\ref{differenceregret}.
To avoid notation confusion, we mention the notation $\hat{\mathbf{p}}_n$ denote the pricing decisions made by algorithm at the end of stage $n-1$, and it will be used in price experiments  at stage $n$; $\mathbf{p}_t$ denote the pricing decisions made by algorithm at time $t$.
\subsection{Limit Point Analysis}
\subsubsection*{Noiseless Case}
\begin{lemma}
Suppose that $\displaystyle{\varepsilon^i_t=0}$, ${\forall i}$ and $t$, and that the sequence $\lbrace \hat{\mathbf{p}}_n \rbrace$, assuming nonzero demand and that the price is away from the limits, generated by CDL converges to a limit point $\tilde{\mathbf{p}}$, which satisfies 
$\displaystyle{\tilde{p}^i=-\frac{\lambda^i(\tilde{\mathbf{p}})}{\nabla_{{p}^i}\lambda^i(\tilde{\mathbf{p}})}}$, and $\tilde{p}^i \in \mathcal{P}^i$.
Then, $\displaystyle{\tilde{\mathbf{p}}}$ is exactly $\mathbf{p}^\ast$. \label{lemma1}
\end{lemma}
\proof{Proof of Lemma \ref{lemma1}}
Without noise, the estimate of the coefficients has a closed-form expression in terms of the rate of demand change with respect to the price increment $\delta_n$:  
\begin{align}
\hat{\beta}_{n+1}^{ii}&=\displaystyle{-\frac{\lambda^i({\hat{p}_n^i+\delta_n,\hat{p}_n^{-i})-\lambda^i({\hat{\mathbf{p}}}_n)}}{\delta_n}}, &\forall i, \label{betaii} \mbox{ and}  \\
\hat{\beta}_{n+1}^{ij}&=
\displaystyle{\frac{\lambda^i({\hat{p}_n^j+\delta_n,\hat{p}_n^{-j})- \lambda^i(\hat{\mathbf{p}}_n)}}{ \delta_n}},  &\forall i,j, i\neq j, \label{betaij}
\end{align}
where $\lambda^i(\hat{p}_n^j+\delta_n,\hat{p}_n^{-j})$ denotes the demand faced by firm $i$ at the price vector whose $j$th element is $\hat{p}_n^j+\delta_n$ and those of others are $\hat{p}_n^{-j}$. It can be inferred from the design of Step 2 in the algorithm that
\begin{align}
\hat{\alpha}_{n+1}^i&=\lambda^i(\hat{\mathbf{p}}_n)+\hat{\beta}_{n+1}^{ii}\hat{p}^i_n-\sum_{j,j\neq i}^F \hat{\beta}_{n+1}^{ij}\hat{p}^j_n, &\forall i ,&\label{alphai} 
\end{align}
and that the closed-form (unconstrained) maximizer of the revenue function is expressed as
\begin{align}
    \hat{p}_{n+1}^i&= \frac{\hat{\alpha}_{n+1}^i+\sum\limits_{j\neq i}^F \hat{\beta}_{n+1}^{ij}\hat{p}_{n+1}^{j}}{2\hat{\beta}_{n+1}^{ii}},&\forall i \label{pi}.            
\end{align}
Suppose that the sequence $\lbrace \hat{\mathbf{p}}_n \rbrace$ converges to limit vector $\tilde{\mathbf{p}}$. 
It must be, therefore, that the sequences $\lbrace \hat{\beta}^{ii}_{n} \rbrace$ and $\lbrace \hat{\beta}_n^{ij} \rbrace$ converge to the directional derivatives of $\lambda^i(\tilde{\mathbf{p}})$ at $\tilde{\mathbf{p}}$}, $\tilde{\beta}^{ii}$ and $\tilde{\beta}^{ij}$, respectively. That is, $$\displaystyle{\tilde{\beta}^{ii}=-\nabla_{{p}^i} \lambda^i(\tilde{\mathbf{p}}) \mbox{ and }
\tilde{\beta}^{ij}=\nabla_{{p}^j} \lambda^i(\tilde{\mathbf{p}})}.$$ Therefore, (\ref{alphai}) implies that $\lbrace \hat{\alpha}_n^{i} \rbrace$ must converge to $$\displaystyle{\tilde{\alpha}^i=\lambda^i(\tilde{\mathbf{p}})-\nabla_{{p}^i} \lambda^i(\tilde{\mathbf{p}})\tilde{p}^i-\sum\limits_{j\neq i}^F \nabla_{{p}^j} \lambda^i(\tilde{\mathbf{p}})\tilde{p}^j.}$$ 
We can derive that
$$\tilde{\alpha}^i-\tilde{\beta}^{ii}\tilde{p}^i+\sum\limits_{j\neq i}^F\tilde{\beta}^{ij}\tilde{p}^j=\lambda^i(\tilde{\mathbf{p}})-\nabla_{{p}^i} \lambda^i(\tilde{\mathbf{p}})\tilde{p}^i-\sum\limits_{j\neq i}^F \nabla_{{p}^j} \lambda^i(\tilde{\mathbf{p}})\tilde{p}^j +\nabla_{{p}^i} \lambda^i(\tilde{\mathbf{p}}) \tilde{p}^i+\sum\limits_{j\neq i}^F\nabla_{{p}^j} \lambda^i(\tilde{\mathbf{p}})\tilde{p}^j=\lambda^i(\tilde{\mathbf{p}}).$$
Thus, $\tilde{p}^i$ must satisfy the following equation:
$$\displaystyle{\tilde{p}^i=\frac{\tilde{\alpha}^i+\sum\limits_{j\neq i}^F \tilde{\beta}^{ij}\tilde{p}^j}{2\tilde{\beta}^{ii}}=\frac{\tilde{p}^i}{2}-\frac{\lambda^i(\tilde{\mathbf{p}})}{2\nabla_{{p}^i}\lambda^i(\tilde{\mathbf{p}})}.}$$

When the mean demand function is $\lambda^i(\mathbf{p})$, $p^i$ satisfies the first-order condition for firm $i$'s revenue maximization problem, which is $\lambda^i(\mathbf{p})+p^i\nabla_{p^i}\lambda^i(\mathbf{p})=0$; when demand function is $\alpha^i-\beta^{ii}{p}^i+\sum\limits_{j\neq i}^F \beta^{ij}p^j$,  $p^i$ satisfies the first-order condition for firm $i$'s revenue maximization problem, which is $\alpha^i-2\beta^{ii}p^i+\sum\limits_{j\neq i}^F \beta^{ij}p^j=0$ . 
The statement above shows that these two first-order conditions have the same output only at $\tilde{p}^i$. 
That is, if the parameters of the linear demand converges to the directional derivative of the true demand, then the first order condition of the revenue functions under these two models are the same.
In a competitive environment, the best response function can be inferred by the first-order condition for each firm's revenue maximization problem. Hence, $\tilde{\mathbf{p}}$ must be the unique clairvoyant NE $\mathbf{p}^\ast$, which can be obtained by solving the concatenated first-order conditions.
\qed
\endproof
The intuition behind this lemma is akin to the gradient descent method, where the decision variables are updated through a first-order approximation. Moreover, the first order approximation of the revenue function $\alpha^i-2\beta^{ii}{p}^i+\sum\limits_{j\neq i}^F \beta^{ij}p^j$ is equivalent to adding the term $-\beta^{ii}{p}^i$ to the linear demand function  $\alpha^i-\beta^{ii}{p}^i+\sum\limits_{j\neq i}^F \beta^{ij}p^j$. And price experimentation aligns the estimation of parameters with the finite difference method. Therefore, we state that the linear model is sufficient for learning when the environment is noiseless.  Next, we argue the Lemma~\ref{lemma1} still holds when the environment is noisy.

\subsubsection*{Analysis with Noises} In fact, noise $\varepsilon_{t}^i$ always exists over the time horizon $T$. Hence, 
\begin{align*}
\hat{\beta}_{n+1}^{ii} &=\displaystyle{-\frac{\sum\limits_{t=t_n+iI_n+1}^{t_n+(i+1)I_n}\bigg(\lambda^i(\hat{p}_n^i+\delta_n,\hat{p}_n^{-i})+\varepsilon^i_{t}\bigg) -\sum\limits_{t=t_n+1}^{t_n+I_n}\bigg(\lambda^i(\hat{\mathbf{p}}_n)+\varepsilon^i_{t}\bigg)}{I_n \delta_n}}\\[10pt]
&=-\nabla_{{p}^i}\lambda^i(\hat{\mathbf{p}}_n)+O(\delta_n)+\frac{1}{\delta_n}\frac{1}{I_n}\left(\sum_{t=t_n+iI_n+1}^{t_n+(i+1)I_n} \varepsilon^i_{t}-\sum\limits_{t=t_n+1}^{t_n+I_n} \varepsilon^i_{t}\right), \; \forall i, \mbox{ and}
\end{align*}
\begin{align*}
\hat{\beta}_{n+1}^{ij} &=\displaystyle{\frac{{\sum\limits_{t=t_n+jI_n+1}^{t_n+(j+1)I_n}\bigg(\lambda^i(\hat{p}_n^j+\delta_n,\hat{p}_n^{-j})+\varepsilon^i_{t}\bigg) -\sum\limits_{t=t_n+1}^{t_n+I_n}\bigg(\lambda^i(\hat{\mathbf{p}}_n)+\varepsilon^i_{t}\bigg)}}{I_n \delta_n}} \\[10pt]
&=\displaystyle{-\nabla_{{p}^j}\lambda^i(\hat{\mathbf{p}}_n)+O(\delta_n)+\frac{1}{\delta_n}\frac{1}{I_n}\left({\sum\limits_{t=t_n+jI_n+1}^{t_n+(j+1)I_n}} \varepsilon^i_{t}-{\sum\limits_{t=t_n+1}^{t_n+I_n} \varepsilon^i_{t}}\right)}, \; \forall i, j,
\end{align*}
where $O(\delta_n)$ denotes the quantity that is, at most, some constant multiplied by $\delta_n$. Derived using Hoeffding’s inequality (formally in E-Companion), an exponential bound shows that $$\displaystyle{I_n^{-1}\left({\sum_{t=t_n+iI_n+1}^{t_n+(i+1)I_n}} \varepsilon^i_{t}-{\sum_{t=t_n+1}^{t_n+I_n} }\varepsilon^i_{t}\right)} \mbox{ and }\displaystyle{I_n^{-1}\left({\sum_{t=t_n+jI_n+1}^{t_n+(j+1)I_n}} \varepsilon^i_{t}-{\sum_{t=t_n+1}^{t_n+I_n} \varepsilon^i_{t}}\right)}$$ may be bounded above by a factor of $(\log (I_n)/I_n)^{1/2}$ with high probability. As $n$ grows, $\delta_n$ and thus $\delta^{-1}_n(\log (I_n)/I_n)^{1/2}$ converge to zero (proved formally in Proposition~\ref{prop2}), then $$\hat{\beta}^{ii}_{n+1} \approx  -\nabla_{{p}^i}\lambda^i(\hat{\mathbf{p}}_n) \mbox{ and } \hat{\beta}^{ij}_{n+1} \approx  \nabla_{{p}^j}\lambda^i(\hat{\mathbf{p}}_n).$$ In particular, when explicitly accounting for noise, the arguments above ensure that the effects of noise vanish as $n$ increases and that the fitted linear model can serve as an estimation of the underlying demand model without being affected by $F$.

However, we show later that when facing competition, the upper bound of revenue regret, derived in the same way as that of one firm, is scaled by $F$ (i.e., Theorem~\ref{theorem3}), the upper bound of revenue difference is scaled by $F^2$ (i.e., Theorem~\ref{theorem2}), and the deviation between the best responses and the clairvoyant NE price is upper bounded by a factor of $F^2I_n^{-{1}/{2}}$ (i.e., Inequality (EC.2)
).

\subsection{Convergence Analysis}\label{convergence}

We clarify that $\hat{\mathbf{p}}_{n+1}$ generated by CDL is obtained by the KKT system (\ref{DDEPKKT}) while the clairvoyant NE $\mathbf{p}^\ast$ is obtained by (\ref{KKT}).
Recall that in Section~2 we have known that each firm's optimal pricing decision at each period can be solved through the best response function $z^i(p^i,p_t^{-i}$). Now for the purpose of analysis, we define
$$\displaystyle{ \breve{\alpha}^i(\mathbf{p})=\lambda^i(\mathbf{p})-\nabla_{p^i} \lambda^i(\mathbf{p})p^i+\sum_{j=1,j\neq i}^F \nabla_{p^j} \lambda^i(\mathbf{p})p^j \mbox{ and }\breve{\beta}^{ij}(\mathbf{p})= \nabla_{p^j} \lambda^i(\mathbf{p})}, \forall i,j,$$
which means that the parameters $\breve{\beta}^{ij}$ form the gradient of the underlying demand curve $\lambda^i(\mathbf{p})$. In this case, the demand shocks $\varepsilon_t^i=0$ for all $t$ and $i$ so there is no \textit{estimation error}. Thus, the corresponding revenue of firm $i$ is 
$$\displaystyle{r^i=p^i\left(\breve{\alpha}^i(\mathbf{p})-\breve{\beta}^{ii}(\mathbf{p})p^i+\sum_{j=1,j\neq i}^F\breve{\beta}^{ij}(\mathbf{p})p^j\right),\forall i.}$$ 
We introduce a best response operator when $\varepsilon_t^i=0$ for all $i$  $$\displaystyle{\breve{z}^i\left(p^{-i} \right)=\argmax_{p^i\in \mathcal{P}^i} p^i\left(\breve{\alpha}^i(\mathbf{p})-\breve{\beta}^{ii}(\mathbf{p})p^i+\sum_{j=1,j\neq i}^F\breve{\beta}^{ij}(\mathbf{p}) p^j\right).}$$ 
{Let $\breve{\mathbf{z}}(\mathbf{p})=\left(\breve{z}^i\left(p^{-i} \right)\right)_i^F$} be the collection of best responses of all firms. The following proposition is based on Theorem~5 in \cite{cachon2006game}.  

\begin{proposition}\label{contraction}
The contraction mapping could
be restated as: for every $z^i$, the matrix $A(\mathbf{p})$ consisting of the derivatives of the best response functions 
$$\displaystyle{A(\mathbf{p})=\begin{bmatrix}
0 & \frac{\partial z^1}{\partial p^2} &\cdots & \frac{\partial z^1}{\partial p^F}\\
\frac{\partial z^2}{\partial p^1}&0 &\cdots &\frac{\partial z^2}{\partial p^F} \\
\cdots &\cdots &\cdots &\cdots \\
\frac{\partial z^F}{\partial p^1} &\frac{\partial z^F}{\partial p^2} &\cdots & 0
\end{bmatrix}},$$
has the largest absolute eigenvalue $\rho(A(\mathbf{p}))<1$ for all $p^i \in \mathcal{P}^i$  (see \cite{horn2012matrix}).
Using the implicit function theorem there (see \cite{bertsekas2016nonlinear}), $\rho(A(\mathbf{p}))<1$ could be restated as the diagonal dominance condition, which is our Assumption~\ref{assumtion1}. (iii). 
\end{proposition}

We have the following proposition for a  \emph{deterministic} demand function as a warm-up. 
\begin{proposition}{}{}
If $\displaystyle{p^{i}=\breve{z}^i\left( p^i,p_t^{-i} \right)}$ for all $i$, then there exists a constant $\displaystyle{\gamma \in (0,1)}$ such that \label{prop1}
$$\Vert \mathbf{p}^{\ast} -\breve{\mathbf{z}}(\hat{\mathbf{p}}_n) \Vert \leq \gamma \left\Vert \mathbf{p}^{\ast} -\hat{\mathbf{p}}_n \right\Vert.$$
\end{proposition}
\proof{Proof of Proposition~\ref{prop1}.}
The parameters $\breve{\alpha}^i(\mathbf{p})$ and $\breve{\beta}^{ij}(\mathbf{p})$ constitute the linear demand function consistent with the underlying demand curve $\lambda^i (\mathbf{p})$ at $\mathbf{p}$ when $\varepsilon^i=0$.  Because of $\breve{\beta}^{ij}(\mathbf{p})= \nabla_{p^j} \lambda^i(\mathbf{p})$ for all $i,j$, then it has that $\sum_{j \neq i}^F \left\vert \breve{\beta}^{ij}(\mathbf{p})\right\vert < \left\vert \breve{\beta}^{ii}(\mathbf{p})  \right\vert$ according to Assumption~\ref{assumtion1}. (iii). Therefore, 
Proposition~\ref{contraction} implies that 
$\breve{\mathbf{z}}(\mathbf{p})$ is a contraction mapping with respect to the norm $\Vert \cdot \Vert$ and there exists a constant $\gamma \in (0,1)$ such that $$\Vert \breve{\mathbf{z}}(\mathbf{p}^\ast)-\breve{\mathbf{z}}(\hat{\mathbf{p}}_n)) \Vert \leq \gamma \Vert \mathbf{p}^\ast -\hat{\mathbf{p}}_n \Vert,$$
According to Lemma 1, we have $\mathbf{p}^{\ast}$ is a fixed point of $\mathbf{z}(\mathbf{p^\ast})$ and $\breve{\mathbf{z}}(\mathbf{p^\ast})$, that is $\mathbf{p}^\ast=\mathbf{z}(\mathbf{p^\ast})=\breve{\mathbf{z}}(\mathbf{p^\ast})$. Concentrating on the equation above, we can obtain that
  $$\Vert \mathbf{p}^{\ast} -\breve{\mathbf{z}}(\hat{\mathbf{p}}_n) \Vert \leq \gamma \left\Vert \mathbf{p}^{\ast} -\hat{\mathbf{p}}_n \right\Vert.$$
where $\hat{\mathbf{p}}_n$ is the equilibrium that we have computed in stage $n$. 
\qed
\endproof

\subsubsection*{Stochastic Demand Functions}
Proposition \ref{prop1} is based on a deterministic  demand function, and the convergence result follows directly from the property of a contraction mapping. Now, we focus on a \emph{stochastic} demand function, and the next lemma shows the uniqueness of $\hat{\mathbf{p}}_n$.


\begin{lemma} \label{lemma2}
Under Assumption 1,$\hat{\mathbf{p}}_{n+1}$ is a unique NE at stage $n+1$ if $\hat{\beta}^{ij}_{n+1}$ converges to $\breve{\beta}^{ij}(\hat{\mathbf{p}}_n)$ for all $i$ and $j$.
\end{lemma}
The proof of this lemma is in (EC.3).
Now we aim to establish the main convergence result.
 The principal technique for establishing the convergence of price in probability is described below. Recall that $\hat{\alpha}^i_{n+1},\hat{\beta}^{ii}_{n+1},\hat{\beta}^{ij}_{n+1}$ are the parameters of the estimated demand model and $\hat{\mathbf{p}}_{n+1}$ is the price vector computed by CDL at the end of stage $n$, then for every stage $n$, we introduce an estimated \emph{best response} operator as
$$\hat{z}^i_n\left(p^{-i}\right)=\argmax_{p^i\in \mathcal{P}^i} p^i\left(\hat{\alpha}_{n+1}^i-\hat{\beta}^{ii}_{n+1}p^i+\sum\limits_{j=1,j\neq i}^{F}\hat{\beta}^{ij}_{n+1}p^j\right).$$
Let $\hat{\mathbf{z}}_n(\mathbf{p})=\left(\hat{z}_n^i\left(p^{-i} \right)\right)_i^F$ denotes the collection of all firms' best responses $\hat{z}^i_n\left(\cdot \right)$ for every stage $n$. Since $\hat{\mathbf{p}}_{n+1}$ is the equilibrium point computed by CDL, $\hat{\mathbf{p}}_{n+1}$ is the fixed point of $\hat{\mathbf{z}}_n(\cdot)$, that is, $\hat{\mathbf{p}}_{n+1}=\hat{\mathbf{z}}_n(\hat{\mathbf{p}}_{n+1})$. Previously, we have conducted the (best  response) operator $\breve{z}^i\left(p^{-i}\right)$  for the case where there is no estimation error,  now we compare the difference between $\breve{z}^i\left(p^{-i}\right)$ and $\hat{z}^i
_n\left(p^{-i}\right)$.
We present the following proposition:

\begin{proposition}{}{}
 \label{prop2} 
At stage $n$, for any given $\hat{p}_{n+1}^i \in \mathcal{P}^i$ generated by CDL, we have  $\hat{\mathbf{p}}_{n+1}=\hat{\mathbf{z}}_n(\hat{\mathbf{p}}_{n+1})$ such that  the operator $\breve{\mathbf{z}}(\hat{\mathbf{p}}_n)$ satisfies, for some positive constant~$K_1$, 
$$\mathbb{E} \left[ \Vert \breve{\mathbf{z}}(\hat{\mathbf{p}}_{n}) -\hat{\mathbf{p}}_{n+1} \Vert^2  \right] \leq F^2K_1 I_n^{-\frac{1}{2}}.$$ 
\end{proposition}
Proposition~\ref{prop2}, which is proved in E-Companion, provides an upper bound for the difference between $\breve{\mathbf{z}}(\hat{\mathbf{p}}_{n})$ and  $\hat{\mathbf{p}}_{n+1}$. The upper bound is proportional to the squared number of competitive firms, $F^2$, and converges to zero as the number of stages increases. To prove this, we first show that for each firm, the difference between the value of the best response function and that of the estimated best response function is bounded above by a constant  $C_n^i$ with high probability, where $C_n^i$ can be viewed as a function of $\delta_n$ (see the proof in EC.4.1).
We define a bad event for firm~$i$ as the difference between these two values exceeding $C_n^i$. If a bad event happens at stage $n$ for one firm, then the vector $\hat{\mathbf{z}}_n(\hat{\mathbf{p}}_{n+1})$ will be affected. Thus, we derive that the norm of the difference vector between the values of the best response functions and those of the estimated best response functions for all firms is upper bounded with a probability of at least $1-({F^2K}/{I_n})$ for a suitable constant $K$, and this probability bound is gradually convergent to $1$ as time grows. Then, we find an upper bound on the expectation of the square norm of the difference vector between the values of the best response functions and those of the estimated best response functions for all firms. The effect of the price experimentation, in which adding $\delta_n$ affects the demand of all participants, should gradually vanish over time. Otherwise, it will cause a poor estimation of the parameters even if demand shocks do not exist. Broadly speaking, the argument above together with Proposition~\ref{prop1} provides a constructive way to design effective price experimentation to obtain the convergence of prices. 
\begin{theorem}{}{}\label{theorem1}Under the assumptions, the prices $\mathbf{p}_t$ at period~$t$, generated by the CDL algorithm, converges to $\mathbf{p}^\ast$ at a rate of $\mathcal{O}(F^2T^{-1/2})$.
\end{theorem}
Note that the NE $\mathbf{p}^\ast$ is the solution that satisfies the KKT system (\ref{KKT}). The proof is stated in EC.4. 
The key idea centers on the following inequality:
 $$ \mathbb{E}\left[\Vert\mathbf{p}^\ast-\hat{\mathbf{p}}_{n+1}\Vert^2\right] \leq \mathbb{E}\left[\left(\Vert \mathbf{p}^\ast-\breve{\mathbf{z}}(\hat{\mathbf{p}}_{n})\Vert +\Vert \breve{\mathbf{z}}(\hat{\mathbf{p}}_{n})-\hat{\mathbf{p}}_{n+1}\Vert \right)^2\right].$$
Combining the results in Propositions \ref{prop1} and \ref{prop2}, we obtain that the expected difference of $\hat{\mathbf{p}}_{n+1}$ and $\mathbf{p}^\ast$ will tend to be $0$ at a rate of $\mathcal{O}(F^2T^{-1/2})$. By the price experiments, where $\delta_n$ is set as $I_n^{-1/4}$, we conclude that $\mathbb{E}\left[\Vert\mathbf{p}_t-\mathbf{p}^\ast\Vert^2\right]$ converges to $0$ at a rate of $\mathcal{O}(F^2T^{-1/2})$.

\subsection{Revenue Difference and Regret}\label{differenceregret}
Theorem~\ref{theorem1} provides the asymptotically equilibrium result for the equilibrium pricing problem: under the designed mechanism, it suffices to show that the linear model guarantees convergence to clairvoyant NE $\mathbf{p}^\ast$. This result offers a fundamental basis for the following analysis. We first analyze the revenue difference, as it may yield useful insights into regret analysis. Recall that the revenue difference denotes the difference in the revenue obtained when all firms set the price to $\mathbf{p}^\ast$ and that obtained when the price is set at $\mathbf{p}_t$. In a monopolistic setting such as that in \cite{besbes2015surprising} and \cite{chen2019coordinating}, through the Taylor expansion, a bound on the regret per period is derived as follows: 
  $$ \vert r(p^\ast)- r(p_t) \vert \leq K (p_t-p^\ast)^2.$$
  The first order term in the Taylor expansion is omitted because its corresponding derivative is equal to $0$ at the optimal price~$p^\ast$, and the second order derivatives are bounded by a positive constant $K$. However, in the competitive setting, the gradient~$\nabla r^i(\mathbf{p}^\ast)$ is a vector that consists of the first order derivatives of $r^i(\mathbf{p}^\ast)$ at each firm's price so that not all elements are equal to $0$, but only the first order derivative with respect to firm~$i$'s price at value $p^{i\ast}$ is equal to $0$. In addition, the second order term must be less than or equal to $0$ due to the concavity of the revenue function. $r^i(\mathbf{p}_t)$ might be greater than $r^i(\mathbf{p}^\ast)$ if the first order term is greater than the absolute value of the second order term, but it is very difficult to judge which term dominates the other. Therefore, we use the revenue difference to measure the difference between $r^i(\mathbf{p}_t)$ and $r^i(\mathbf{p}^\ast)$. We obtain that
   \begin{align*}
    \left\vert r^i(\mathbf{p}^\ast)- r^i(\mathbf{p}_t) \right\vert & \leq \left\vert \nabla r^i(\mathbf{p}^\ast)^T( \mathbf{p}_t-\mathbf{p}^\ast)-\frac{1}{2} (\mathbf{p}_t-\mathbf{p}^\ast)^T\nabla^2 r^i(\mathbf{p}^\ast)(\mathbf{p}_t-\mathbf{p}^\ast)\right\vert\\
    &\leq K^\prime \left\Vert \mathbf{p}_t-\mathbf{p}^\ast \right\Vert+K^{\prime\prime}\left\Vert \mathbf{p}_t-\mathbf{p}^\ast \right\Vert^2,
\end{align*}
where $K^\prime$ and $K^{\prime\prime}$ are positive constants, and the last inequality uses the triangle inequality (i.e., $\vert a+b\vert \leq \vert a\vert+\vert b \vert$ for any real numbers $a$ and $b$). 

\begin{theorem}{}{}
Under the defined assumptions, the vector sequence~$ \displaystyle{\lbrace \mathbf{p}_t:t\geq 1 \rbrace}$ (compared with the clairvoyant Nash equililbrium~$\mathbf{p}^*$) satisfies
$$\mathbb{E} \left[\sum_{t=1}^T\left\vert r^i(p^{i\ast},p^{-i\ast})-r^i(p^i_t,p^{-i}_t)\right\vert \right]\leq F^2 K_6 T^{\frac{3}{4}},\quad  \forall i=1,\cdots,F,$$
for some positive constant $K_6$, $T \geq 2$, and $F \geq 2$. \label{theorem2} 
\end{theorem}
Theorem~\ref{theorem2}, whose proof is in EC.5,
implies that, by using this pricing policy, even if the underlying demand function is unknown, the difference between the revenue $r^i(p^{i\ast},p^{-i\ast})$ at the clairvoyant $\mathbf{p}^\ast$ and the revenue $r^i(p^i_t,p^{-i}_t)$ at $\mathbf{p}_t$ obtained by CDL is at most $F^2 K_6  T^{3/4}$, in which $F$ represents the number of participating firms. In other words, the revenue difference per period converges to zero. To prove this theorem, we use the convergence results in Theorem~\ref{theorem1}, which showed the converge rates of $\mathbf{p}_t$ to $\mathbf{p}^\ast$. We find a bound on the convergence rate of each term to show an upper bound on the revenue difference. The first order term $K^\prime\left\Vert \mathbf{p}-\mathbf{p}^\ast \right\Vert$ is the key factor raising an upper bound to order $\mathcal{O}(T^{3/4})$. When $F$ equals 1, meaning that there are no competitors in the market, the first order term is equal to $0$ and an upper bound is of order $\mathcal{O}(T^{1/2})$. An upper bound on the revenue difference ($\sum_{t=1}^T\vert r^i(\mathbf{p}_t)-r^i(\mathbf{p}^\ast) \vert$) of order $\mathcal{O}(T^{3/4})$ is also shown based on the convergence rate of $\vert r^i(\mathbf{p}_t)-r^i(\mathbf{p}^\ast) \vert$ in \cite{meylahn2022learning}. (See Lemma~4 of \cite{meylahn2022learning}.)


Theorem~\ref{theorem2} gives us an upper bound on revenue difference, and we now analyze the revenue regret. It is known that the price $p_t^i$ of firm $i$ made by CDL is not necessarily the optimal price that truly maximize $i$'s own revenue. Such a true optimal price $p_t^{i\ast}$ is derived from the true best response function with respect to the true demand. 
Suppose that firm $i$ realizes the true demand function at some time $T$. In such a case, the firm would regret making the price decision~$p_t^i$ rather than $p_t^{i\ast}$ during the past period~$t=1,\ldots,T$. 
Hence, we obtain the following theorem regarding the regret bound whose proof is in EC.6:
\begin{theorem}{}{}
Under the defined assumptions, the sequence~$\{p^i_t: t\geq 1\}$ (compared with the sequence~$\lbrace p_t^{i\ast}: t\geq 1 \rbrace$) satisfies
$$\displaystyle{\mathbb{E} \left[\sum\limits_{t=1}^T\Big[r^i(p^{i\ast}_t,p^{-i}_t)-r^i(p^i_t,p^{-i}_t)\Big] \right]\leq FK_7  T^{\frac{1}{2}},\quad  \forall i=1,\cdots,F,}$$
for some positive constant $K_7$,  $T \geq 2$ and $F \geq 2$. \label{theorem3}
\end{theorem}

The key ideas to prove this theorem are the following: we divide the price experimentation (Step~1) into three parts and analyze the difference between $p_t^{i\ast}$ and $p_t^i$ of each part separately. The first part contains the first $I_n$ periods, where no adjustment are made on $p_t^i$. The second part contains the periods where firm $i$ change its price. The third part contains the periods where the firms except firm $i$ change their prices. Then, we obtain the convergence rate of $p_t^i$ to $p_t^{i\ast}$. From this, the regret bound is derived. 

Theorem~\ref{theorem3} shows that after time $T$, the regret between making pricing at $p_t^i$ and making pricing at $p_t^{i\ast}$ is at most $FK_7 T^{1/2}$, which is also proportional to the number of firms. As we have stated, the revenue differences have an inferior theoretical bound compared to those obtained based on the regret of a single firm due to the non-omission of the first  order term.  In our problem, the revenue difference measures the performance of the algorithm against a static benchmark in which all firms knew their demand functions in advance  (that is, the equilibrium $\mathbf{p}^\ast$). The regret measures the performance of the algorithm against a dynamic benchmark in which the firm knows the others' pricing decisions at each period $t$ and its own demand function in advance. We denote by $p_t^{i\ast}$ the optimal price of firm $i$ with the knowledge of the underlying demand. The revenue used to compute regret is the revenue obtained at $p_t^{i\ast}$, which is the true revenue maximizer at time $t$.
The optimal pricing $p_t^{i\ast}$ is virtual and is, therefore,  not useful in capturing sales information. Thus, firm $i$ cannot predict the subsequent pricing of other firms, and the convergence of $p_t$ to $p_t^{i\ast}$ is unrelated to the convergence of $p_t^{-i}$ to $p_t^{-i\ast}$. 
In fact, when $p^{i\ast}_t$ no longer changes, it means that $p^{i\ast}_t$ becomes a fixed point of the best response function, $z^i$, indicating that pricing decisions have reached the clairvoyant NE $\mathbf{p}^\ast.$ 

\remark{The worst-case regret as specified in \cite{besbes2009dynamic}, \cite{broder2012dynamic} is the performance measure of a pricing policy when the nature ``picks'' the worst possible demand function, and they provided a problem instance to show a lower bound on the worst-case regret, which illustrates the limitation on the performance of \emph{any} admissible pricing policy. Their regret lower bound demonstrated that the regret upper bound was tight, meaning that the lower bound and upper bound have the same order in terms of $T$ asymptotically. An information-theoretic regret lower bound was established in \cite{keskin2014dynamic}. They showed that the regret lower bound has an order of $\sqrt{T}$ asymptotically. To our best knowledge, regret lower bound on a similar measure for the competitive setting has not appeared in the literature. Since our result matches the aforementioned bounds, our conjecture is that our regret upper bound is tight, and we leave establishing regret lower bound to future work.}

\subsection{Extension: Partially-Clairvoyant Model and Results}\label{flawedmodel}
Given the results above, a question arises as to \emph{whether the regret and difference bounds still hold when some firms have knowledge of the underlying demand curves $\lambda ^i(\cdot)$ and the distributions of demand shocks $\varepsilon^i$}.

\cite{cooper2015learning} considered a linear model in which decision makers ignore the impact of competitors, that is, the so-called \emph{flawed} model $\alpha^i-\beta^{ii} p^i$ in which the prices of the other firms do not appear in an estimated demand curve. An important property in their work is the demand consistency property. They considered a limit point $\tilde{\mathbf{p}}$ which is the clairvoyant NE and assumed that all firms use these flawed estimated demand curves. The estimated demand and the expected demand at $\tilde{\mathbf{p}}$ are equal, that is, $\alpha^i-\beta^i\tilde{p}^{i\ast}=\mathbb{E}[D^i(\tilde{\mathbf{p}})]$. However, using this flawed model may lead the firms to converge to some potential limit points that are \emph{not} Nash equilibria. Thus, the flawed estimated demand curves \emph{cannot} be used directly for our purpose. 

Suppose that $F^\prime$ firms do not know the underlying demand curve. We thus define the partially-flawed approximated demand curves under the partially-clairvoyant model for these firms as  \[\alpha^i-\beta^{ii}p^i+\sum\limits_{j=1,j\neq i}^{F^\prime} \beta^{ij} p^j,\] in which the effects of the firms with the knowledge of underlying demand curves do not appear in this function.  Let $k$ denote the index of the firms with known demand, $k=F^\prime+1,\cdots,F$. The price experimentation step now only involves $F^\prime$ firms.
\begin{itemize}
    \item Step 1. Setting prices: Firm $k$'s pricing $p_t^k$, for $k=F^\prime+1,\cdots,F$, is:
\begin{equation*}
p_t^k=\hat{p}_n^k  \;\; \;\;\forall t=t_{n}+1,\cdots,t_{n}+(F^\prime+1)I_n.
\end{equation*}
\end{itemize}
 
In Step~2, $F^\prime$ firms keep learning the demand function via a linear regression and the linear approximation model is $\alpha^i-\beta^{ii}p^i+\sum\limits_{j\neq i}^{F^\prime} \beta^{ij}p^j$. The prices of firms with knowledge of the underlying demand functions are not involved in this linear approximation model. Step~3 of the modified CDL is given as follows:
\begin{itemize}
\item Step 3. Computing the equilibrium: Solve the following two systems simultaneously: 

\begin{equation}
\begin{array}{cl}
\begin{array}{c}
\hat{\alpha}^i_{n+1}-2\hat{\beta}^{ii}_{n+1}p^i+\sum\limits_{j,j\neq i}^{F^\prime} \hat{\beta}^{ij}_{n+1}p^{j} 
\end{array}
-\mu^{i,1}+\mu^{i,2}=0, & \forall i,\\[6pt]
\mu^{i,1}\geq 0, \mu^{i,1}\cdot\left(p^i-p^{i,h}\right)=0,p^i-p^{i,h} \leq 0 &\forall i,\\[6pt]
\mu^{i,2}\geq 0, \mu^{i,2}\cdot\left(p^{i,l}-p^i\right)=0,p^{i,l}-p^i \leq 0 &\forall i,
\end{array}\label{eq:i}
\end{equation}
and
\begin{equation}
\begin{array}{cl}
\begin{array}{cc}
\lambda^k(\mathbf{p})+p^k \nabla_{p^k} \lambda^k(\mathbf{p})+\mathbb{E}(\varepsilon^k)
\end{array}
-\mu^{k,1}+\mu^{k,2}=0 & \forall k,\\[6pt]
\mu^{k,1}\geq 0, \mu^{k,1}\cdot\left(p^k-p^{k,h}\right)=0, p^k-p^{k,h}\leq 0 &\forall k,\\[6pt]
\mu^{k,2}\geq 0, \mu^{k,2}\cdot\left(p^{k,l}-p^k\right)=0, p^{k,l}-p^k\leq 0 &\forall k.
\end{array}\label{eq:k}
\end{equation}
Then, the prices for each firm $\hat{p}^i_{n+1}$ and $\hat{p}^k_{n+1}$ are set to the solution of the above simultaneous equations \eqref{eq:i} and \eqref{eq:k}. 
\end{itemize}

In the above algorithm, the linear model is different from that described in Section~\ref{model}. We justify the choice of such partially-clairvoyant approximated demand curves in the following. If firms without knowledge of demand curves use the full estimated demand curves, the number of price-demand combinations in one stage should be $(F+1)$ to determine the values of $F+1$ unknown parameters. However, the clairvoyant firms have no need to estimate the demand curves and are thus not motivated to pursue price experimentation.\footnote{According to the diagonal dominance condition, if the firms without the knowledge of demand curves attend the price experimentation but the clairvoyant firms do not, then the revenue of the clairvoyant firms will increase.} Therefore, the price experimentation only generates $(F^\prime+1)$ price-demand combinations. If the firms without the knowledge of demand curves consider the full  estimated demand curves (i.e., $\alpha^i-\beta^{ii}p^i+\sum_{j,j\neq i}^F\beta^{ij}p^j$) while the clairvoyant firms do not attend the price experimentation, the result of experimentation is under-determined equations, resulting in inaccurate  parameters estimation. The algorithm requires that clairvoyant firms update prices based on their true best response functions at the end of each stage simultaneously. The prices of the clairvoyant firms would have no influence on the pricing decisions of the firms without knowledge of the demand functions while the latter's pricing decisions affect the former's. 

\subsubsection*{Analyses}
An analogy to Lemma~\ref{lemma1} in this partially-clairvoyant setting is as follows:

\begin{lemma}
Suppose that $\varepsilon^i_t=0$, $\forall i=1,\cdots,F^\prime$ and $t$, and that the sequence $\lbrace \hat{\mathbf{p}}_n \rbrace$ generated by CDL converges to a limit point $\tilde{\mathbf{p}}$, which satisfies 
$\displaystyle{\tilde{p}^i=-\frac{\lambda^i(\tilde{\mathbf{p}})}{\nabla_{\tilde{p}^i}\lambda^i(\tilde{\mathbf{p}})}}$, and $\tilde{p}^i \in \mathcal{P}^i$.
Then, $\displaystyle{\tilde{\mathbf{p}}}$ is exactly $\mathbf{p}^\ast$. \label{lemma3}
\end{lemma}
Lemma~\ref{lemma3}, which is proved in EC.7,
implies that although the impact of the clairvoyant firms' prices are omitted in the  linear approximation function $\alpha^i-\beta^{ii}p^i+\sum\limits_{j=1,j\neq i}^{F^\prime} \beta^{ij} p^j$, it is sufficient for our learning. 
For firm~$i$ without knowing the underlying demand, totally unaware of the clairvoyant firm $k$'s price, the reaction from firm $k$ is included in the parameter $\alpha^i$. In addition, the clairvoyant firm $k$ is still very sensitive to price changes by its competitors. From the best response function's perspective, $\hat{p}^i_{n+1}$ is set by \[\arg \max p^i\left(\hat{\alpha}^i_{n+1}-\hat{\beta}^{ii}_{n+1}p^i+\sum\limits_{j,j\neq i}^{F^\prime} \hat{\beta}^{ij}_{n+1}\hat{p}_{n+1}^j\right).\] 
Compared with the previous arguments in Section~\ref{analysis}, the main difference is the parameter $\hat{\alpha}^i_{n+1}$ in both models. Hence, during the estimation cycle, an obvious result can be used to describe the difference between them, namely $\sum\limits_{k=F^\prime+1}^{F}\nabla_{p^k}\lambda^i(\mathbf{p})p^k$. 
\begin{corollary}
Under the assumptions,  supposing $F^\prime$ firms do not know the underlying demand function, the prices~$\mathbf{p}_t$ at period~$t$ converges to $\mathbf{p}^\ast$ at a rate of $\mathcal{O}(F^{\prime 2}T^{-1/2})$.
\end{corollary}
Note that the NE $\mathbf{p}^\ast$ is the solution that satisfies the KKT system (\ref{eq:i}) and (\ref{eq:k}). We turn our attention to regret analysis. To demonstrate the influence of the firms with knowledge of the underlying demand functions, we show the following result whose proof is in EC.8.

\begin{theorem}{}{}
Under the previously stated assumptions, supposing that $F^\prime$ firms do not know the underlying demand function, the vector sequence~$\left\lbrace \mathbf{p}_t:t\geq 1 \right\rbrace$ (compared with the clairvoyant Nash equilibrium~$\mathbf{p}^*$) 
satisfies
$$\displaystyle{\mathbb{E} \left[\sum\limits_{t=1}^T\Big[\vert {r^i(p^{i\ast},p^{-i\ast})}-r^i(p^i_t,p^{-i}_t)\vert\Big] \right]\leq F^{\prime 2} K_8 T^{\frac{3}{4}}},\quad  \forall i=1,\cdots,F^\prime,$$
for some positive constant $K_8$, $T \geq 2$ and $F \geq 2,1\leq F^\prime \leq F$. \label{theorem4}
\end{theorem}
The cumulative difference over $T$ periods between the revenues ${r^i(p^{i\ast},p^{-i\ast})}$ and $r^i(p_t^i,p^{-i}_t)$ is at most $F^{\prime2} K_8  T^{\frac{3}{4}}$. Comparing this to the results contained in Theorem~\ref{theorem2}, one can see an apparent difference: $F^2$ changes into $F^{\prime2}$, which implies that the revenue difference is dependent on the number of the firms without knowledge of the demand curves while the upper bound on this difference is applicable to any firm with or without the knowledge of the underlying demand function. It stands to reason that if a firm does not know the demand function, it must make efforts to learn the demand function so the equilibrium will gradually converge to the clairvoyant NE $\mathbf{p}^\ast$. If any firm does not reach its price at the NE $\mathbf{p}^\ast$, there is no way for all firms to stabilize pricing. Similar to the results in Theorem~\ref{theorem3}, we derive a further theorem proved in EC.9.

\begin{theorem}{}{}
Under the previously stated assumptions, the sequence~$\{p^i_t: t\geq 1\}$ (compared with the sequence~$ \left\lbrace p_t^{i\ast} :t\geq 1 \right\rbrace$) for each $i$ of the firms that do not know the demand curves satisfies
$$\displaystyle{\mathbb{E} \left[\sum\limits_{t=1}^T\Big[r^i(p^{i\ast}_t,p^{-i}_t)-r^i(p^i_t,p^{-i}_t)\Big] \right]\leq F^\prime K_9   T^{\frac{1}{2}}},\quad \forall i=1,\cdots,F^\prime,$$
for some positive constant $K_9$, $T \geq 2$, $F \geq 2$, and $1\leq F^\prime \leq F$.\\
the sequence~$\{p^i_t: t\geq 1\}$ (compared with the sequence~$ \left\lbrace p_t^{i\ast} :t\geq 1 \right\rbrace$) for each $i$ of the firms who know the demand curves satisfies
$$\displaystyle{\mathbb{E} \left[\sum\limits_{t=1}^T\Big[r^i(p^{i\ast}_t,p^{-i}_t)-r^i(p^i_t,p^{-i}_t)\Big] \right]\leq  K_{10}   T^{\frac{1}{2}},\quad \forall i=F^\prime+1,\cdots,F,}$$
for some positive constant $K_{10}$, $T \geq 2$, $F \geq 2$, and $1\leq F^\prime \leq F$. \label{theorem5}
\end{theorem}
It should be emphasized that the two kinds of firms have different regret bounds as shown in Theorem~\ref{theorem5}. The price decisions of the clairvoyant firms always follows the best responses according to the true demand curves, and there is no estimation error to affect the pricing decisions. However, their pricing decisions cannot affect firms without knowledge of the demand curves. 
Theorem~\ref{theorem5} shows that the upper bound on the regret for firms without knowledge of the demand curves is still associated with the number of those firms while the upper bound on the regret for the clairvoyant firms is not.

\section{Numerical Experiments}\label{numericalexperiments}

We examine the performance of the CDL algorithm and present the numerical results related to the regret. We consider two demand curve environments for $\lambda^i(\mathbf{p})$:\\
1): Linear models: for each firm $i$, $\lambda^i(\mathbf{p})=\alpha^i-\beta^{ii}p^i+\sum_{j\neq i}^F p^j$, $\alpha^i\in [\underline{\alpha}^i,\overline{\alpha}^i],\beta^{ii} \in [\underline{\beta}^{ii},\overline{\beta}^{ii}], \mbox{ and } \beta^{ij} \in [\underline{\beta}^{ij},\overline{\beta}^{ij}]$, where $[\underline{\alpha}^i,\overline{\alpha}^i]= [3,5]$, $[\underline{\beta}^{ii},\overline{\beta}^{ii}]= [0.8,0.9]$, $[\underline{\beta}^{ij},\overline{\beta}^{ij}]= [0.6,0.7]$ when $F=2$, $[\underline{\beta}^{ij},\overline{\beta}^{ij}]= [0.3,0.35]$ when $F=3$, $[\underline{\beta}^{ij},\overline{\beta}^{ij}]= [0.2,0.23]$ when $F=4$, and $[\underline{\beta}^{ij},\overline{\beta}^{ij}]= [0.15,0.175]$ when $F=5$,\\
2): Multinomial logit models: for each firm $i$, $\displaystyle{\lambda^i(\mathbf{p})=\frac{\exp^{\alpha^i-\beta^ip^i}}{1+\sum\limits_{i=1}^F \exp^{\alpha^i-\beta^ip^i}}, \alpha^i\in [\underline{\alpha}^i,\overline{\alpha}^i], \mbox{ and } \beta^i \in [\underline{\beta}^i,\overline{\beta}^i],}$ where $[\underline{\alpha}^i,\overline{\alpha}^i]= [3,4]$ and $[\underline{\beta}^i,\overline{\beta}^i]= [0.4,0.5]$,\\
 The random error of $\varepsilon_t^i$ is assumed to follow a normal distribution with mean 0 and variance $\sigma^2$. Owing to the properties of the multinomial logit model, each firm's demand does not exceed 1. Moreover, as the number of firms increases, setting the standard deviation at a fixed value leads to an unrealistic experimental setting. Hence, for the multinomial logit model, we set the standard deviation as the ratio ($\kappa$) of the average demand ($\bar{\lambda}_F$) of $F$ firms, i.e., $\sigma=\kappa\cdot \bar{\lambda}_F$. This varies with the number of firms participating in the price experimentation. But for the linear models, the standard deviation does not change with the number of firms.
 
For each specification of the above settings, we randomly draw 100 instances for the parameters and a sample path of demand shock for each firm. $\alpha^i$, $\beta^{ii}$ and $\beta^{ij}$ are drawn according to a uniform distribution on $[\underline{\alpha}^i,\overline{\alpha}^i]$ , $[\underline{\beta}^{ii},\overline{\beta}^{ii}]$ and  $[\underline{\beta}^{ij},\overline{\beta}^{ij}]$.
In all experiments, we set an initial batch size $I_0=1$ and price interval $\mathcal{P}^i=[p^{i,l},p^{i,h}]=[0,6]$. The initial price $\hat{p}_1^i$ is randomly selected from $\mathcal{P}^i$ and the value of $v$ is set at 2, for each firm $i=1,\cdots,F$.
To present the result of the convergence of regret, for each period $T$ we compute the fraction of revenue loss as follows:
$$\frac{\mathcal{R}^i(\pi^i,T)}{\sum_{t=1}^T p_t^{i\ast}\mathbb{E}[D_t^i(p_t^{i\ast},p_t^{-i})]}.$$
Meanwhile, we compute the fraction of revenue difference defined by:
$$\frac{\mathcal{D}^i(\pi^i,T)}{\sum_{t=1}^T p_t^{i\ast}\mathbb{E}[D_t^i(\mathbf{p}^\ast)]}.$$

 Figures~2(a) and 2(b) show respectively the averaged fraction of revenue loss and averaged fraction of revenue difference of a firm who has no knowledge of the underlying demand when the underlying demand curve follows a linear model with $\sigma$ set to $0.16$. Figures~2(c) and 2(d) show the corresponding results of a multinomial logit model with $\kappa$ set to $5\%$. The sampled standard deviations  for the averaged fractions of revenue loss and revenue difference are below $0.05\%$ over 100 repeated runs.
Several important observations can be made: as time goes on, the averaged fractions of revenue loss and revenue difference of each firm converge to zero; as the number of firms increases, the averaged fractions of revenue loss and revenue difference of each firm increase significantly. Table~\ref{tablea} reports the averaged fraction of revenue loss and averaged fraction of revenue difference of a firm that does not know the underlying demand under different standard deviations {when the underlying demand is linear, multinomial logit, exponential, and semi-log, respectively \footnote{ For the exponential models, $\alpha^i$, $\beta^{ii}$ and $\beta^{ij}$ are uniformly drawn from $[0.8,1.2]$, $[0.25,0.5]$ and $[0.1,0.2]$. For the semi-log models, $\alpha^i$, $\beta^{ii}$ and $\beta^{ij}$ are uniformly drawn from $[0.8,1]$, $[0.6,0.8]$ and $[0.2,0.4]$.}.}
 We observe that as the standard deviation of demand shocks increases, the averaged fractions of revenue loss and revenue difference become larger. Furthermore, when the underlying demand curve follows a linear model and the firm's estimation model is well specified, the numerical results are better than those for a multinomial logit model, an exponential model, and a semi-log model.

\begin{figure}[!ht]
		\centering
  \includegraphics[scale=1.8]{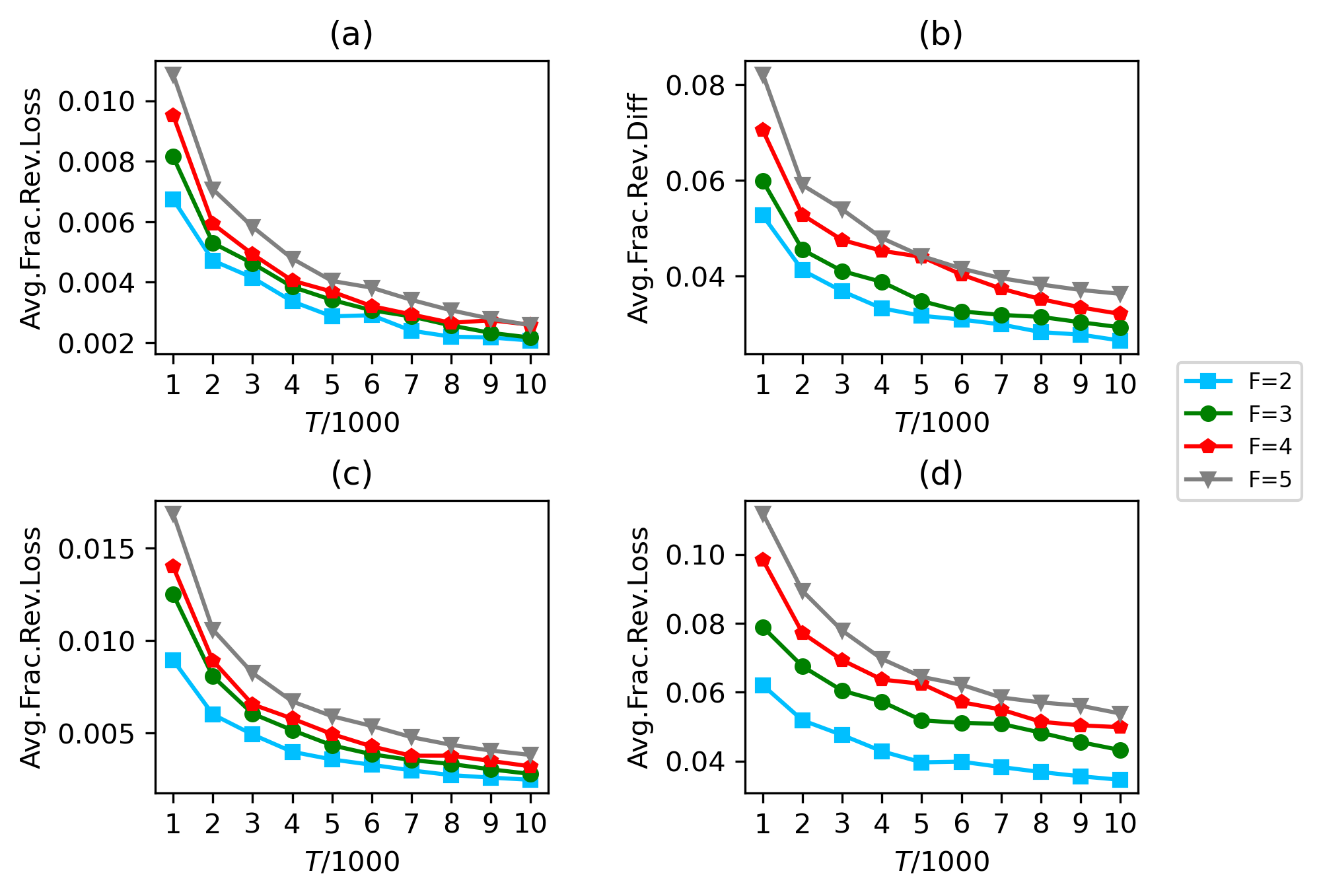}
  \label{figure2}
 	\caption{The averaged fraction of revenue loss and averaged fraction of revenue difference of a firm who has no knowledge of the underlying demand for different numbers $F$.	(The upper panels are the results of linear underlying model and the lower panels are the results of multinomial logit underlying model.)}
\end{figure}
\begin{table}[!ht]
\caption{The averaged fraction of revenue loss and averaged fraction of revenue difference of a firm who has no knowledge of the underlying demand for different standard deviations $\sigma$. ($F=2$).\label{tablea}}
\centering
\begin{threeparttable}
\begin{tabular}{ccccccccccc}
 \hline\toprule&&&& Periods $T$ &&&&& \\\hline
       Avg. Frac. Revenue Loss (Linear) & & 2000& &  4000  & &6000 & &8000 &&10000\\\hline
      $\sigma=0.16 $& &  0.0051	
&  & 0.0033 & & 0.0028 && 0.0025 && 0.0021\\
$\sigma=0.32 $ & &  0.0139	
&  & 0.0092 & & 0.0075 && 0.0062 && 0.0058 \\
$\sigma=0.48 $ & &  0.0315	
&  & 0.0228 & & 0.0195 && 0.0163 && 0.0145 \\\hline
Avg. Frac. Revenue Loss (MNL) & & 2000& &  4000  & &6000 & &8000 &&10000\\\hline
      $\kappa=5\% $& &  0.0062	
&  & 0.0041 & & 0.0034 && 0.0028 && 0.0024\\
$\kappa=10\%  $ & &  0.0166	
&  & 0.0115 & & 0.0092 && 0.0079 && 0.0072 \\
$\kappa=15\%  $ & &  0.0332	
&  & 0.0254 & & 0.0207 && 0.0183 && 0.0169 \\\hline
Avg. Frac. Revenue Loss (Exp) & & 2000& &  4000  & &6000 & &8000 &&10000\\\hline
      $\kappa=5\% $& &  0.0101	
&  & 0.0074 & & 0.0061 && 0.0049 && 0.0044\\
$\kappa=10\%  $ & &  0.0199	
&  & 0.0135 & & 0.0113 && 0.0093 && 0.0082 \\
$\kappa=15\%  $ & &  0.0355	
&  & 0.0272 & & 0.0205 && 0.0173 && 0.0152 \\\hline
Avg. Frac. Revenue Loss (Semi-log) & & 2000& &  4000  & &6000 & &8000 &&10000\\\hline
      $\kappa=5\% $& &  0.0103	
&  & 0.0076 & & 0.0068 && 0.0056 && 0.0045\\
$\kappa=10\%  $ & &  0.0133	
&  & 0.0095 & & 0.0085 && 0.0070 && 0.0065 \\
$\kappa=15\%  $ & &  0.0197	
&  & 0.0142 & & 0.0121 && 0.0105 && 0.0094 \\\hline
 Avg. Frac. Revenue Difference (Linear) & & 2000& &  4000  & &6000 & &8000 &&10000\\\hline
$\sigma=0.16 $& &  0.0432	
&  & 0.0372 & & 0.0338 && 0.0301 && 0.0281\\
$\sigma=0.32 $ & &  0.0637	
&  & 0.0503 & & 0.0449 && 0.0402 && 0.0382 \\
$\sigma=0.48 $ & &  0.1038	
&  & 0.0820 & & 0.0713 && 0.0625 && 0.0562 \\\hline
 Avg. Frac. Revenue Difference (MNL) & & 2000& &  4000  & &6000 & &8000 &&10000\\\hline
       $\kappa=5\% $& &  0.0543
&  & 0.0443 & & 0.0410 && 0.0369 && 0.0352\\
$\kappa=10\%  $ & &  0.0835	
&  & 0.0657 & & 0.0562 && 0.0518 && 0.0491\\
$\kappa=15\%  $ & &  0.1112
&  & 0.0909 & & 0.0794 && 0.0729 && 0.0665\\\hline
Avg. Frac. Revenue Difference (Exp) & & 2000& &  4000  & &6000 & &8000 &&10000\\\hline
       $\kappa=5\% $& &  0.0437
&  & 0.0357 & & 0.0321 && 0.0292 && 0.0270\\
$\kappa=10\%  $ & &  0.0636
&  & 0.0518 & & 0.0468 && 0.0417 && 0.0394\\
$\kappa=15\%  $ & &  0.0808
&  & 0.0649 & & 0.0591 && 0.0525 && 0.0598\\\hline
Avg. Frac. Revenue Difference (Semi-log) & & 2000& &  4000  & &6000 & &8000 &&10000\\\hline
       $\kappa=5\% $& &  0.0787
&  & 0.0642 & & 0.0550 && 0.0505 && 0.0483\\
$\kappa=10\%  $ & &  0.0836
&  & 0.0696 & & 0.0643 && 0.0546 && 0.0515\\
$\kappa=15\%  $ & &  0.1013
&  & 0.0847 & & 0.0781 && 0.0702 && 0.0687\\\hline
\end{tabular}
\end{threeparttable}
\end{table}

We perform additional analysis to validate our theorems. We first evaluate the numerical results of the averaged fraction of revenue loss and averaged fraction of revenue difference with respect to the time horizon to confirm the theoretical results. According to Theorem~\ref{theorem2}, $\mathcal{D}^i(\pi^i,T) \leq K_6F^2T^{3/4}$ implies that $\mathcal{D}^i(\pi^i,T)/T \leq K_6F^2T^{-1/4}$; and by Theorem~\ref{theorem3}, $\mathcal{R}^i(\pi^i,T) \leq K_7FT^{1/2}$ implies that $\mathcal{R}^i(\pi^i,T)/T \leq K_7FT^{-1/2}$.  This implies that the log of the average fraction of revenue loss should be a linear function of $\log T$ with a slope of approximately $-1/2$, and the log of the average fraction of revenue difference should fit a linear function of $\log T$ with a slope of approximately $-1/4$. When the underlying demand curve is a linear model, Figure 3(a) shows that the log of the averaged fraction of revenue loss of a firm that has no knowledge of the underlying demand asymptotically approaches a line with a slope of $-0.5$ when plotted against $\log T$ (base $e$). Figure 3(b) presents that the logarithm of the averaged fraction of revenue difference of a firm who has no knowledge of the underlying demand fits a line with a slope of $-0.25$. When the underlying demand curve is multinomial logit model, Figures 3(c) and 3(d) show that the log of the averaged fractions of revenue loss and revenue loss fit a line with slopes of $-0.5$ and $-0.25$, respectively. We can also approximate the averaged fractions of revenue loss and revenue difference as functions of $F$. 
Table~\ref{table1} reports the averaged fractions of revenue loss and revenue difference of firm 1 when $F=2,3,4$, and $5$. For example, if we select the results of $T=2000$, the averaged fraction of revenue loss of firm 1 can be approximated by $0.005F+0.0043$, and the averaged fraction of revenue difference can be approximated by $0.0023F^2+0.03$. Therefore, these results numerically confirm Theorem~\ref{theorem2} and Theorem~\ref{theorem3}.

\begin{figure}[!ht]
		\centering
 \includegraphics[scale=1.8]{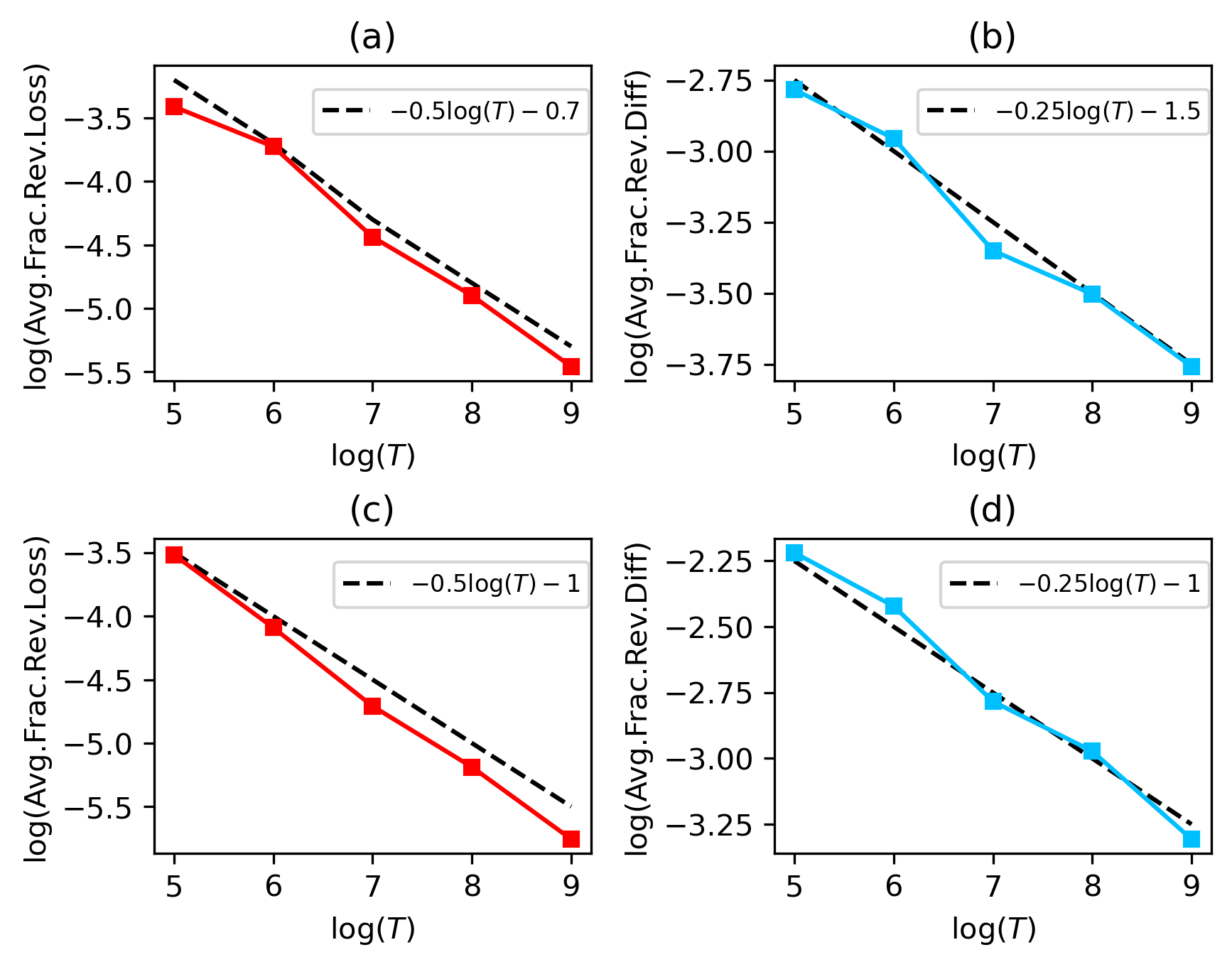}		\label{figure3}
 	\caption{The log-log plots of averaged fractions of revenue loss and revenue difference of a firm who has no knowledge of the underlying demand as a function of $\log (T)$ when $F=2$ and the underlying demand curves are linear models (upper panels) or multinomial logit models (lower panels).}
\end{figure}

\begin{table}[!ht]
\caption{The averaged fractions of revenue loss and revenue difference of a firm who has no knowledge of the underlying demand in terms of different number $F$ (Linear model).\label{table1}}
\centering
\begin{threeparttable}
\begin{tabular}{ccccccccccc}
 \hline\toprule&&&& Periods $T$ &&&&& \\\hline
       Avg. Frac. Revenue Loss & & 2000& &  4000  & &6000 & &8000 &&10000\\\hline
       F=2& &  0.0051	
&  & 0.0033 & & 0.0028 && 0.0025 && 0.0021\\
F=3 & &  0.0057	
&  & 0.0038 & & 0.0031 && 0.0026 && 0.0022 \\
F=4 & &  0.0062	
&  & 0.0043 & & 0.0034 && 0.0028 && 0.0024 \\
F=5 & &  0.0069
&  & 0.0046 & & 0.0036 && 0.0029 && 0.0025 \\\hline
 Avg. Frac. Revenue Difference & & 2000& &  4000  & &6000 & &8000 &&10000\\\hline
F=2& &  0.0432	
&  & 0.0372 & & 0.0338 && 0.0301 && 0.0281\\
F=3 & &  0.0462
&  & 0.0382 & & 0.0349 && 0.0321 && 0.0295 \\
F=4 & &  0.0586
&  & 0.0456 & & 0.0401 && 0.0360 && 0.0329 \\
F=5 & &  0.0853
&  & 0.0603 & & 0.0484 && 0.0431&& 0.0390 \\\hline
\end{tabular}
\end{threeparttable}
\end{table}

For the partially-clairvoyant model, we test the performance of the modified CDL algorithm with different numbers of firms that possesses knowledge of the underlying demand ($F^\prime=2,3$, and $4$). We consider the underlying demand curves to be multinomial logit models with the same settings as the CDL algorithm and set $\kappa=5\%$. Figure 4 shows the results of the averaged fractions of revenue loss and revenue difference. we can see that the averaged fractions of revenue loss and revenue difference are decreasing over time. Figures 4(a) and 4(b) report the results of the firms who have no knowledge of underlying demand and show that the averaged fractions of revenue loss and revenue difference of these firms increase as $F^\prime$ increases. Figures 4(c) and 4(d) report the results of the clairvoyant firms and show that the averaged fraction of revenue difference of the clairvoyant firms increases as $F^\prime$ increases, but the averaged fraction of revenue loss does not. 
\begin{figure}[!ht]
		\centering
 \includegraphics[scale=1.8]{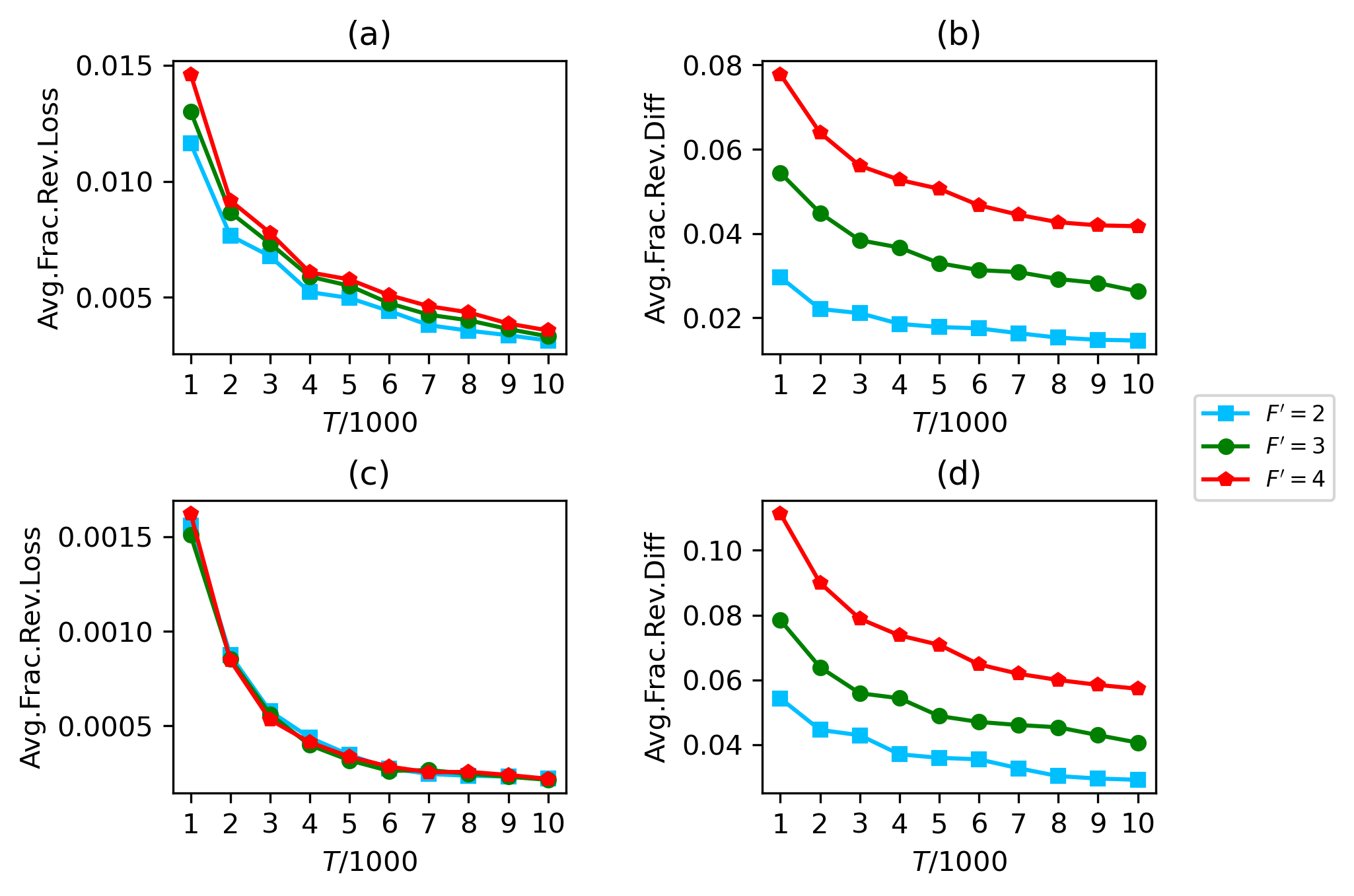}		
 	\caption{The averaged fractions of revenue loss and revenue difference of a firm who has no knowledge of the underlying demand curve (upper panels) and a firm who knows the underlying demand curve (lower panels). The total number of firms $F=5$ and the underlying demand curves are multinomial logit models.}
\end{figure}

We also verify the results of revenue loss and revenue difference with the theoretical results in Theorem~\ref{theorem4} and Theorem~\ref{theorem5}.  Figures 5(a) and 5(b) show that the log of the averaged fraction of revenue loss of a firm who has no knowledge of the underlying demand curve has a linear relationship with respect to $\log T$ with a slope close to $-0.5$ and the log of the averaged fraction of revenue difference of this firm fits a line with a slope of $-0.25$. For the results of a clairvoyant firm, Figure 5(c) shows that if the log of its averaged fraction of revenue loss fits a line, its slope is approximately $-0.7$, which is lower than $-0.5$. This is because a clairvoyant firm does not participate in the price experiments, while the non-clairvoyant firms do. Meanwhile, a clairvoyant firm's revenue will be increased since the non-clairvoyant firms raise the prices for the price experiments. Figure 5(d) presents that the log of a clairvoyant firm's averaged fraction of revenue loss fits a line with a slope of $-0.25$. Additionally, we find that the averaged fraction of revenue loss of a firm that has no knowledge of the underlying demand can be described as a linear function with respect to $F^\prime$. For example, if we select the results of averaged fraction of revenue loss when $T=2000$ in Table \ref{table2}, the corresponding linear function is $0.004F^\prime+0.0075$. Similarly, if we use a quadratic function with respect to $F^\prime$ to approximate the averaged fraction of revenue difference, based on the results of averaged fraction of revenue difference when $T=2000$ in Table \ref{table2}, the quadratic function is $0.003F^{\prime2}+0.015$. Table \ref{table3} shows the results of a clairvoyant firm, and we observe that the averaged fraction of revenue loss of the clairvoyant firm does not depend on $F^\prime$, but the results of averaged fraction of revenue difference in Table \ref{table3} can be modeled as $0.003F^{\prime2}+0.03$ when $T=2000$. Thus, these numerical results validate our Theorem~\ref{theorem4} and Theorem~\ref{theorem5}.

\begin{figure}[!ht]
		\centering
 \includegraphics[scale=1.8]{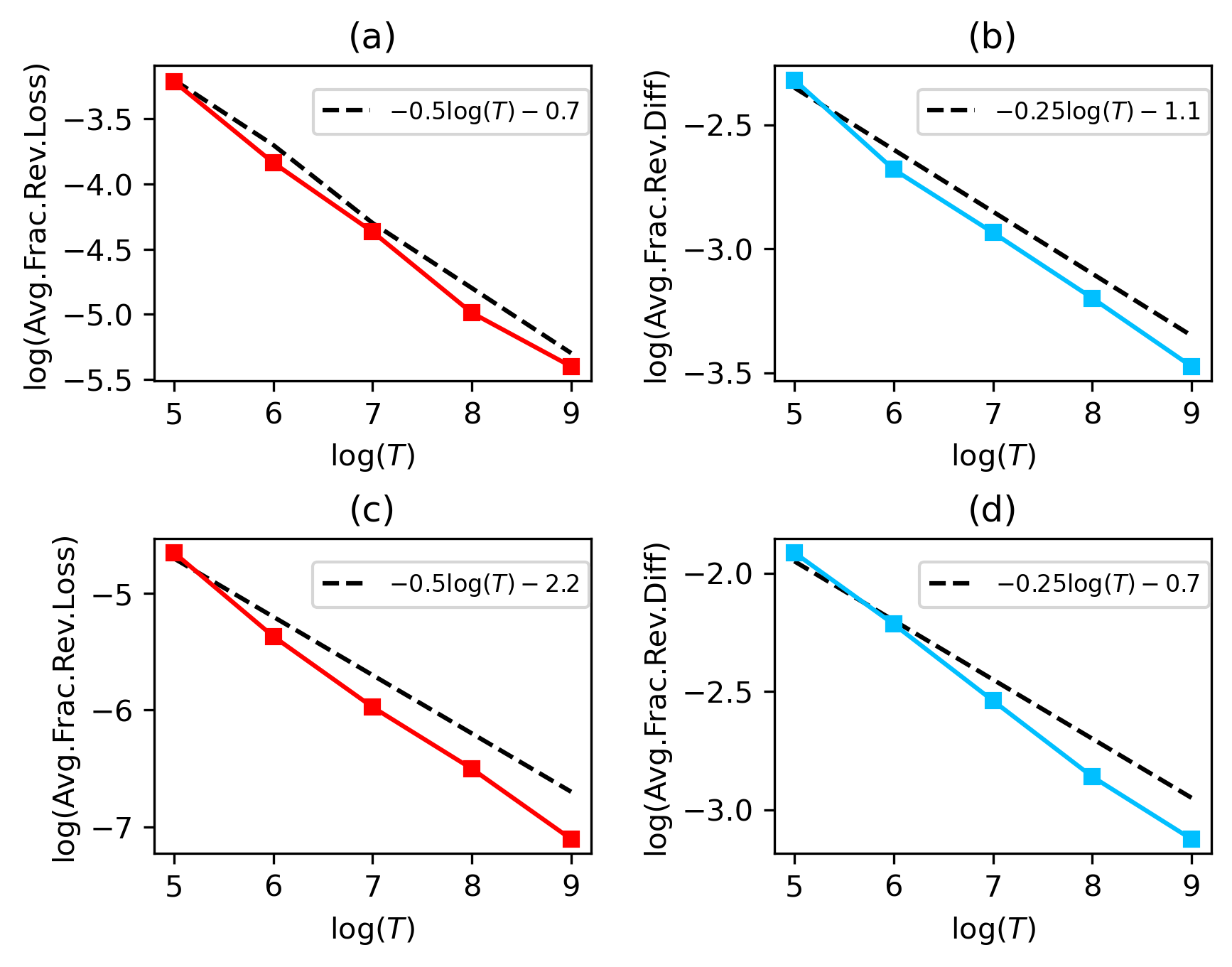}		
 	\caption{The log-log plots of averaged fractions of revenue loss and revenue difference of the firms who have no knowledge of the underlying demand as a function of $\log (T)$ (upper panels) and those of the clairvoyant firms as a function of $\log (T)$ (lower panels). $F^\prime=3, F=5$, and the underlying demand curves are multinomial logit models. }
\end{figure}
\begin{table}[!ht]
\caption{The averaged fractions of revenue loss and revenue difference of a firm who has no knowledge of the underlying demand in terms of different numbers $F^\prime$ (multinomial logit model).\label{table2}}
\centering
\begin{threeparttable}
\begin{tabular}{ccccccccccc}
\hline\toprule&&&& Periods $T$ &&&&& \\\hline
       Avg. Frac. Revenue Loss & & 2000& &  4000  & &6000 & &8000 &&10000\\\hline
       $F^\prime$=2& &  0.0082	
&  & 0.0055 & & 0.0047 && 0.0037 && 0.0032\\
$F^\prime$=3 & &  0.0086	
&  & 0.0057 & & 0.0049 && 0.0039 && 0.0033 \\
$F^\prime$=4 & &  0.0091	
&  & 0.0060 & & 0.0049 && 0.0040 && 0.0035 \\\hline
 Avg. Frac. Revenue Difference & & 2000& &  4000  & &6000 & &8000 &&10000\\\hline
$F^\prime$=2& &  0.0251	
&  & 0.0195 & & 0.0175 && 0.0153 && 0.0146\\
$F^\prime$=3 & &  0.0428	
&  & 0.0346 & & 0.0303 && 0.0272 && 0.0243 \\
$F^\prime$=4 & &  0.0639	
&  & 0.0527 & & 0.0467 && 0.0426 && 0.0417 \\\hline
\end{tabular}
\end{threeparttable}
\end{table}
\begin{table}[!ht]
\caption{The averaged fractions of revenue loss and revenue difference of a clairvoyant firm in terms of different numbers $F^\prime$ (multinomial logit model).\label{table3}}
\centering
\begin{threeparttable}
\begin{tabular}{ccccccccccc}
 \hline\toprule&&&& Periods $T$ &&&&& \\\hline
       Avg. Frac. Revenue Loss & & 2000& &  4000  & &6000 & &8000 &&10000\\\hline
       $F^\prime$=2& &  0.0008	
&  & 0.0004 & & 0.0003 && 0.0002 && 0.0001\\
$F^\prime$=3 & &  0.0008	
&  & 0.0004 & & 0.0004 && 0.0002 && 0.0002\\
$F^\prime$=4 & &  0.0008	
&  & 0.0004 & & 0.0003 && 0.0002 && 0.0002\\\hline
 Avg. Frac. Revenue Difference & & 2000& &  4000  & &6000 & &8000 &&10000\\\hline
$F^\prime$=2& &  0.0445
&  & 0.0370 & & 0.0354 && 0.0302 && 0.0291\\
$F^\prime$=3 & &  0.0638	
&  & 0.054 & & 0.0470 && 0.0453 && 0.0406 \\
$F^\prime$=4 & &  0.0899	
&  & 0.0737 & & 0.0647&& 0.0599&& 0.0572 \\\hline
\end{tabular}
\end{threeparttable}
\end{table}

Finally, we consider a case where the firm deviates from the prescribed modified CDL algorithm. More precisely, the clairvoyant firms change the prices every period instead of at the end of each stage.  Table~\ref{table4} shows the averaged fraction of revenue loss and averaged fraction of revenue difference when $F=5$, $F^\prime=3$, and that the clairvoyant firms change prices every period based on the best response function. We observe that changing prices per period can make the clairvoyant firms generate more revenue loss and revenue difference compared to adjusting prices at the end of each stage, i.e., compared to the results in Table~\ref{table3}. We also find that in this case, the results of the non-clairvoyant firms, who do not know the underlying demand curves, have no significant changes compared to the results in Table~\ref{table2}.

\begin{table}[!ht]
\caption{A comparative case where the clairvoyant firms deviate from the modified CDL algorithm when $F=5$ and $F^\prime=3$. The averaged fraction of revenue loss and averaged fraction of revenue difference as the clairvoyant firms change prices every period (multinomial logit model).\label{table4}}
\centering
\begin{threeparttable}
\begin{tabular}{ccccccccccc}
 \hline\toprule&&&& Periods $T$ &&&&& \\\hline
       Avg. Frac. Revenue Loss & & 2000& &  4000  & &6000 & &8000 &&10000\\\hline
       clairvoyant firms & &  0.0030	
&  & 0.0028 & & 0.0024 && 0.0023 && 0.0021\\
non-clairvoyant firms & &  0.0082	
&  & 0.0055 & & 0.0041 && 0.0036 && 0.0031\\\hline
Avg. Frac. Revenue Difference & & 2000& &  4000  & &6000 & &8000 &&10000\\\hline
clairvoyant firms & &  0.0066	
&  & 0.0053 & & 0.0045 && 0.0043 && 0.0039\\
non-clairvoyant firms & &  0.0432	
&& 0.0353&  & 0.0306 & & 0.0293 && 0.0265\\\hline
\end{tabular}
\end{threeparttable}
\end{table}

\section{Conclusions}
In this paper, we have studied a non-cooperative game among the firms each seeking to maximize revenue. The underlying demand-price information is not known \textit{a priori} for each firm and the demand response is determined by consumers and competitors. Based on past observations, estimated demand models are used to make pricing decisions. The equilibrium pricing algorithm is developed for such a system and is combined with estimation and optimization cycles. The estimation cycle uses a linear regression model to estimate the linear approximation of the underlying demand model at a previously given price. We show that this converges to the equilibrium pricing if the algorithm is designed properly. 
The role of the algorithm is to guarantee the convergence of pricing decisions. {We also conduct an analysis of the regret. Regret is defined, in our case, as a comparison with the optimal price that the firm could make when clairvoyant information is available. 
And we show that each firm's regret divided by $T$ converges to $0$ at the rate of $\mathcal{O}(T^{-1/2})$.

We use a linear approximation to estimate the unknown demand function. We analyze the linear demand model in two general scenarios: one being when all competing firms have unknown demand functions, the other being when a subset of demand functions is unknown. In the first scenario, the linear demand model captures the impact of all competitors. In the second scenario (partially clairvoyant model), only the impact of competitors with unknown demand functions is captured. We show that the linear demand model in both scenarios ensures that pricing decisions converge to the true equilibrium pricing as a result of the algorithm operation cycle. The analysis and numerical results show that a regret upper bound of firms without knowledge of the demand curve is associated with the number of such firms, while a regret upper bound of the firms with knowledge of the demand curve is not. We have not found literature establishing regret lower bounds for such problems, and due to the specificity of our definition of regret, some classical methods may not be suitable for deriving regret lower bounds, so we leave the task of establishing theoretical regret lower bounds for future research. 


A potential follow-up topic would be to consider the effect of product differentiation. Recently some research has focused on dynamic pricing problems where the underlying demand curve involves product features, this scenario where the competing firms sell differentiated products is left to future work. Another viable direction for future research would be to combine the price competition with the inventory problem. Some works have also studied dynamic inventory and pricing models in a competitive environment, such as \cite{bernstein2004dynamic}. 
Another exciting area for future work would be the development of an algorithm that considers other types of functions as the approximation functions.

\ACKNOWLEDGMENT{%
This research is supported by Ministry of Science and Technology of Taiwan under grants MOST 107-2221-E-007-074-MY3, MOST 108-2221-E-009-053-MY2, MOST 110-2221-E-007-108-MY3, and MOST 111-2410-H-A49-022-MY2. We would like to thank Yue Dai (Fudan University), Chiao-Wei Li (NTHU), Wen-Ying Huang (NTHU), and Ling-Wei Wang for useful discussions.
}

\bibliographystyle{apalike}  
\bibliography{POM/ref1}

\end{document}